\documentclass[paper=letter, fontsize=12pt]{article}

\usepackage[english]{babel} % English language/hyphenation
\usepackage{amsmath,amsfonts,amsthm} % Math packages
\usepackage[utf8]{inputenc}
\usepackage{float}
\usepackage{lipsum} % Package to generate dummy text throughout this template
\usepackage{blindtext}
\usepackage{graphicx} 
\usepackage{caption}
\usepackage[sc]{mathpazo} % Use the Palatino font
\usepackage[T1]{fontenc} % Use 8-bit encoding that has 256 glyphs
\usepackage{microtype} % Slightly tweak font spacing for aesthetics
\usepackage[hmarginratio=1:1,top=32mm,columnsep=20pt]{geometry} % Document margins
\usepackage{multicol} % Used for the two-column layout of the document
\usepackage{booktabs} % Horizontal rules in tables
\usepackage{float} % Required for tables and figures in the multi-column environment - they need to be placed in specific locations with the [H] (e.g. \begin{table}[H])
\usepackage{hyperref} % For hyperlinks in the PDF
\usepackage{lettrine} % The lettrine is the first enlarged letter at the beginning of the text
\usepackage{paralist} % Used for the compactitem environment which makes bullet points with less space between them
\usepackage{abstract} % Allows abstract customization
\usepackage{titlesec} % Allows customization of titles
\usepackage{natbib, url}
\usepackage[caption=false]{subfig}   % load for sub figures

\renewcommand\thesection{\Roman{section}} % Roman numerals for the sections
\renewcommand\thesubsection{\Roman{subsection}} % Roman numerals for subsections

\titleformat{\section}[block]{\large\scshape\centering}{\thesection.}{1em}{} % Change the look of the section titles
\titleformat{\subsection}[block]{\large}{\thesubsection.}{1em}{} % Change the look of the section titles
\usepackage{fancyhdr} % Headers and footers
\usepackage[bottom]{footmisc}

\newtheorem{theorem}{Theorem}[section]
\newtheorem{lemma}[theorem]{Lemma}

\newtheorem{corollary}[theorem]{Corollary}

\newtheorem{remark}[theorem]{Remark}

\numberwithin{equation}{section}

\newcommand{\given}{\,|\,}
\newcommand{\T}{\top}
\newcommand{\ba}{\mathbf{a}}

\newcommand{\bs}{\mathbf{s}}
\newcommand{\bS}{\mathbf{S}}
\newcommand{\calS}{\mathcal{S}}
\newcommand{\bC}{\mathbf{C}}
\newcommand{\bY}{\mathbf{Y}}
\newcommand{\by}{\mathbf{y}}
\newcommand{\bX}{\mathbf{X}}
\newcommand{\bx}{\mathbf{x}}
\newcommand{\bZ}{\mathbf{Z}}
\newcommand{\bz}{\mathbf{z}}
\newcommand{\bh}{\mathbf{h}}
\newcommand{\bV}{\mathbf{V}}
\newcommand{\bA}{\mathbf{A}}

\newcommand{\bF}{\mathbf{F}}
\newcommand{\fb}{\mathbf{f}}

\newcommand{\bu}{\mathbf{u}}
\newcommand{\bU}{\mathbf{U}}
\newcommand{\calU}{{\cal U}}
\newcommand{\bH}{\mathbf{H}}
\newcommand{\bM}{\mathbf{M}}
\newcommand{\calM}{{\cal M}}
\newcommand{\calR}{{\cal R}}
\newcommand{\bI}{\mathbf{I}}
\newcommand{\bD}{\mathbf{D}}
\newcommand{\bL}{\mathbf{L}}
\newcommand{\bR}{\mathbf{R}}

\newcommand{\bv}{\mathbf{v}}

\newcommand{\bP}{\mathbf{P}}
\newcommand{\bG}{\mathbf{G}}
\newcommand{\bmu}{\boldsymbol{\mu}}
\newcommand{\bSigma}{\boldsymbol{\Sigma}}
\newcommand{\beps}{\boldsymbol{\epsilon}}
\newcommand{\bbeta}{\boldsymbol{\beta}}
\newcommand{\bLambda}{\boldsymbol{\Lambda}}
\newcommand{\bomega}{\boldsymbol{\omega}}
\newcommand{\bOmega}{\boldsymbol{\Omega}}
\newcommand{\bgamma}{\boldsymbol{\gamma}}
\newcommand{\etab}{\boldsymbol{\eta}}
\newcommand{\bPsi}{\boldsymbol{\Psi}}
\newcommand{\brho}{\boldsymbol{\rho}}
\newcommand{\bkappa}{\boldsymbol{\mathcal{K}}}
\newcommand{\blambda}{\boldsymbol{\lambda}}
\newcommand{\bzero}{\mathbf{0}}

\graphicspath{{./pics/}}

\title{\vspace{-15mm}\fontsize{24pt}{10pt}\selectfont\textbf{Spatial Factor Modeling: A Bayesian Matrix-Normal Approach for Misaligned Data}} % Article title

\author{
\large
{\textsc{Lu Zhang}}\\[2mm]
{\textsc{UCLA Department of Biostatistics }}\\[2mm]
\normalsize \href{mailto:lu.zhang@ucla.edu}{Lu.Zhang@ucla.edu}\\[2mm] 
\large
{\textsc{Sudipto Banerjee}}\\[2mm]
{\textsc{UCLA Department of Biostatistics }}\\[2mm]
\normalsize \href{mailto:sudipto@ucla.edu}{sudipto@ucla.edu}\\[2mm] % Your email address
}
\date{June 1, 2020}
\providecommand{\keywords}[1]{\textbf{\textit{Key words:}} #1}
\begin{document}
\maketitle % Insert title
\thispagestyle{fancy} % All pages have headers and footers

\label{firstpage}

\begin{abstract}
Multivariate spatially-oriented data sets are prevalent in the environmental and physical sciences. Scientists seek to jointly model multiple variables, each indexed by a spatial location, to capture any underlying spatial association for each variable and associations among the different dependent variables. Multivariate latent spatial process models have proved effective in driving statistical inference and rendering better predictive inference at arbitrary locations for the spatial process. High-dimensional multivariate spatial data, which is the theme of this article, refers to data sets where the number of spatial locations and the number of spatially dependent variables is very large. The field has witnessed substantial developments in scalable models for univariate spatial processes, but such methods for multivariate spatial processes, especially when the number of outcomes are moderately large, are limited in comparison. Here, we extend scalable modeling strategies for a single process to multivariate processes. We pursue Bayesian inference which is attractive for full uncertainty quantification of the latent spatial process. Our approach exploits distribution theory for the Matrix-Normal distribution, which we use to construct scalable versions of a hierarchical linear model of coregionalization (LMC) and spatial factor models that deliver inference over a high-dimensional parameter space including the latent spatial process. We illustrate the computational and inferential benefits of our algorithms over competing methods using simulation studies and an analysis of a massive vegetation index data set.
\end{abstract}

\newpage

\section{Introduction}\label{chp4sec: Intro}
{Statistical modeling for multiple spatially-oriented data are required to capture underlying spatial associations in each variable and accounting for inherent associations among the different variables. As an example, to which we return later, consider a set of spatially indexed spectral variables for vegetation activity on the land. Such variables exhibit strong spatial dependence as customarily exhibited through plots of spatial variograms and other exploratory maps. In addition, the variables are assumed to be associated with each other because of shared physical processes that manifest through the observations.}

{Modeling each variable separately captures the spatial distribution of that variable independent of other variables. Such analysis ignores associations among the variables and can impair prediction or interpolation \citep[see, e.g.,][]{chiles2009geostatistics, wackernagel2003, gelban10, cressie2015statistics}. Each of the aforementioned works provide ample evidence, theoretical and empirical, in favor of joint modeling of multiple spatially indexed variables. Joint modeling, or multivariate spatial analysis is especially pertinent in the presence of spatial misalignment, where not all variables have been observed over the same set of locations. For example, suppose $Y(\bs)$ is Normalized Difference Vegetation Index (NDVI) and $X(\bs)$ is red reflectance. If location $\bs_0$ has yielded a measurement for $X(\bs_0)$ but not for $Y(\bs_0)$, then optimal imputation of $Y(\bs_0)$ should proceed from $p(Y(\bs_0)\given \bY, \bX)$, where $\bY$ and $\bX$ comprise all measurements on $Y(\bs)$ and $X(\bs)$. If the processes $Y()$ and $X()$ are modeled as independent, then the predictive distribution $p(Y(\bs_0)\given \bY, \bX) = p(Y(\bs_0)\given \bY)$ and will not exploit the possible predictive information present in $X(\bs_0)$ for $Y(\bs_0)$. This specific issue has also been discussed, with examples, in \cite{bangel2002}.} 

Joint modeling is driven by vector-valued latent spatial stochastic processes, such as a multivariate Gaussian process. These are specified with matrix-valued cross-covariance functions \citep[see, e.g.,][and references therein]{genton2015cross, salvana2020nonstationary, le2006statistical} that models pairwise associations at distinct locations. Theoretical properties of cross-covariances are well established, but practical modeling implications and computational efficiency require specific considerations depending upon the application \citep[see, e.g. ][]{%brown1994multivariate,
le1997bayesian, sun1998assessment, le2001spatial, gamerman2004multivariate, schmidt2003bayesian, banerjee2014hierarchical}.     

High-dimensional multivariate spatial models will deal with a large number of dependent variables over a massive number of locations. While analyzing massive spatial and spatial-temporal databases have received attention \citep[see, e.g.,][]{sunligenton11,banerjee2017high,heaton2019case,zhang2020high}, the bulk of methods has focused on one or very few (two or three) spatially dependent variables and often have to rely upon restrictive assumptions that preclude full inference on the latent process. With larger numbers of dependent variables, modeling the cross-covariance becomes challenging. Even for stationary cross-covariance functions, where we assume that the associations among the variables do not change over space and the spatial association for each variable depends only on the difference of two positions, matters become computationally challenging. 
  
{This manuscript builds upon the popular linear models of coregionalization \citep{bourgault1991multivariable, goulard1992linear, wackernagel2003, gelnonstat, chiles2009geostatistics, genton2015cross}. Our contributions include: (i) developing a hierarchical model with a Matrix-Normal distribution as a prior for an unknown linear transformation on latent spatial processes; (ii) extending classes of spatial factor models for spatially misaligned data; (iii) accounting for multiple outcomes over very large number of locations.} Spatial factor models have been explored by \cite{Wang2003}, \cite{Lopes2008}, \cite{renbanerjee2013} and \cite{taylor2019spatial}. \cite{Lopes2008} provides an extensive discussion on how hierarchical models emerged from dynamic factor models. \cite{renbanerjee2013} proposed low-rank specifications for spatially-varying factors to achieve dimension reduction, but such low-rank specifications tend to over-smooth the latent process from massive data sets containing millions of locations. More recently, \cite{taylor2019spatial} consider Nearest-Neighbor Gaussian process \citep{datta16} for spatial factors with the usual constrained loading matrices in non-spatial factor models. These are more restrictive than needed for identifying spatially correlated factors \citep[see, e.g.][]{renbanerjee2013}.       

We develop our modeling framework in Section~\ref{sec: Multivariate_Spatial}. %Detailed algorithms for implementing models in Section~\ref{sec: Multivariate_Spatial} with Nearest Neighbor Gaussian Process (NNGP) \citep{datta16} are given in supplementary material. 
{Section~\ref{sec: Theory} presents some theoretical results about posterior consistency for the proposed models.} Simulation studies for exploring the performance of proposed models are summarized in Section~\ref{sec: simulation}. Section~\ref{sec: real_data_analy} presents an application to remote-sensed vegetation analysis on land surfaces. We conclude with some discussion in Section~\ref{sec: summary}. 

\section{Multivariate spatial processes}\label{sec: Multivariate_Spatial}
Let $\bz(\bs) = (z_1(\bs), \ldots, z_q(\bs))^\top$ be a $q \times 1$ stochastic process, where each $z_i(\bs)$ is a real-valued random variable at location $\bs \in \mathcal{D}\subseteq \Re^d$. The process is specified by its mean $\mbox{E}[z_i(\bs)]=\mu_i(\bs)$ and, customarily, second-order stationary covariances $C_{ij}(\bh)= \mbox{Cov}\{z_i(\bs),z_j(\bs+\bh)\}$ for $i,j=1,2,\ldots,q$. These covariances define the matrix-valued $q\times q$ cross-covariance function $\bC(\bh) = \{C_{ij}(\bh)\}$ with $(i,j)$-th entry $C_{ij}(\bh)$. While there is no loss of generality in assuming the process mean to be zero by absorbing the mean into a separate regression component in the model, as we will do here, modeling the cross-covariance function requires care. From its definition, $\bC(\bh)$ need not be symmetric, but must satisfy $\bC(\bh)^{\top} = \bC(-\bh)$. Also, since $\mbox{var}\{\sum_{i}^n\ba_i^{\top}\bz(\bs_i)\} \geq 0$ for any set of finite locations $\bs_1,\bs_2,\ldots,\bs_n\in \mathcal{D}$ and any set of constant vectors $\ba_1,\ba_2,\ldots,\ba_n\in \Re^q$, we have $\sum_{i,j=1}^n \ba_i^{\top}\bC(\bs_i-\bs_j)\ba_i \geq 0$. \cite{genton2015cross} provide a comprehensive review of cross-covariance functions. 

%While theoretical characterizations rely upon spectral theory and are useful in understanding the local behavior of random fields, 
Perhaps the most widely used approach for constructing multivariate random fields is the linear model of coregionalization (LMC). {This hinges on invertible linear maps of independent spatial processes yielding valid spatial processes}. If $\fb(\bs) = (f_1(\bs),f_2(\bs),\ldots,f_K(\bs))^{\top}$ is a $K\times 1$ vector of independent spatial processes so that $\mbox{cov}\{f_i(\bs),f_j(\bs')\}=0$ for all $i\neq j$ and any two locations $\bs$ and $\bs'$ (same or distinct), then LMC \citep{bourgault1991multivariable} specifies %the $q\times 1$ process %$\bz(\bs) = \bLambda^{\top}\fb(\bs)$
%\textcolor{blue}{
\begin{equation}\label{eq: SLMC}
    \bz(\bs) = \sum_{k = 1}^K \blambda_k f_k(\bs) = \bLambda^\top \fb(\bs)\;,
\end{equation}
where {$\bz(\bs)$ is $q\times 1$}, $\bLambda$ is $K\times q$, $\blambda_k^{\top}$ is the $k$-th row of $\bLambda$ and each $f_k(\bs)$ is an independent Gaussian process with correlation function $\rho_{\psi_k}(\cdot, \cdot)$ with parameters $\psi_k$.
The cross-covariance for %$\fb(\bs)$ is diagonal and for 
$\bz(\bs)$ %is $\bC_z(\bh) = \bLambda^{\top}\bC_f(\bh)\bLambda$. %} 
%This 
yields non-degenerate process-realizations whenever %\textcolor{blue}{
 $K \geq q$ %}
and $\bLambda$ is nonsingular. %\textcolor{blue}{
To achieve dimension reduction in the number of variables, we restrict $K < q$ so we have non-degenerate realizations in a $K$ dimensional sub-space. %}

%The key question, then, is how to model $\bLambda$, whose rows determine the subspace where the factors are mapped.   
%We follow the developments in \citet{bourgault1991multivariable}, which we call the simplified LMC (or SLMC). The multivariate random field $\bz(\bs)$ is a linear combination of $K$ independent univariate random fields
%\begin{equation}\label{eq: SLMC}
%    \bz(\bs) = \sum_{k = 1}^K \blambda_k f_k(\bs) = \bLambda^\top \fb(\bs)\;,
%\end{equation}
%where $\blambda_k$ is the $k$-th row of $\bLambda$ and each $f_k(\bs)$ is an independent Gaussian process with correlation function $\rho_{\psi_k}(\cdot, \cdot)$, where $\psi_k$ are parameters in the correlation function. The corresponding cross-covariance function for $\bz(\bs)$ is $\bC(\bs, \bs')= \sum_{k = 1}^K \rho_{\psi_k}(\bs, \bs') \blambda_{k}\blambda_{k}^\top$. We call model \eqref{eq: SLMC} a simplified linear model of coregionalization (SLMC). 
\cite{schmidt2003bayesian} propose multivariate spatial processes through a hierarchical spatial conditional model, whereupon $\bLambda^\top$ in (\ref{eq: SLMC}) is a $q \times q$ lower triangular matrix. 
%%\cite{zhang2007maximum} developed an EM algorithm for estimating the maximum likelihood estimation of the GLMC model, whereas \cite{schmidt2003bayesian} proposed a Bayesian approach for LMC when the loading matrix $\bLambda$ in (\ref{eq: LMC}) is a lower triangular matrix.  
%The LMC \eqref{eq: SLMC} connects univariate spatial models to multivariate spatial models using linear transformations of univariate processes. 
Other variants of LMC \citep[e.g.][]{goulard1992linear} can also be recast as \eqref{eq: SLMC} using linear algebra. The flexibility offered in modeling $\bLambda$ is appealing and, in particular, can accrue computational benefits in high-dimensional settings. %Other approaches for building cross-covariance functions such as convolutions, latent dimensions, Mat\'ern cross-covariances and other methods reviewed in \citet{genton2015cross} do not necessarily provide the flexibility and scalability we seek. 
Hence, we build upon (\ref{eq: SLMC}). 

\subsection{A Bayesian LMC factor model (BLMC)}\label{subsec: Bayesian_SLMC}
Let $\by(\bs) = (y_1(\bs), \ldots, y_q(\bs))^\top \in \mathbb{R}^q$ denote the $q\times 1$ vector of dependent outcomes in location $\bs \in \mathcal{D} \subset \mathbb{R}^d$, $\bx(\bs) = (x_1(\bs), \ldots, x_p(\bs))^\top \in \mathbb{R}^p$ be the corresponding explanatory variables, and $\bbeta$ be a $p \times q$ regression coefficient matrix in the multivariate spatial model 
\begin{equation}\label{eq: SLMC_model_element}
    \by(\bs) = \bbeta^\top \bx(\bs) + \bLambda^\top \fb(\bs)+ \beps(\bs)\; , \; \bs \in \mathcal{D}\; ,
\end{equation}
where the latent process $\bLambda^\top \fb(\bs)$ is an LMC as described above.
%Here, $\bLambda$ is the $K \times q$ loading matrix and $\bLambda^\top \fb(\bs)$ is the latent process. 
%\subsection{Bayesian SLMC using MNIW}\label{sec: SLMC_model_def}
%Though conjugate multivariate spatial models have unparalleled advantages in obtaining quick inference for a large-scale spatial dataset, the restriction of a separable design of the underlying dependence structure of the data can be too strong for many applications. Besides, the conjugacy is only guaranteed when fixing a set of hyper-parameters in the covariance functions, and can easily be violated when having misalignment. Here we introduce a Bayesian framework for the more versatile multivariate spatial models, the simplified linear model of coregionalization (SLMC) model. With modest conditions, the Bayesian SLMC model can provide full Bayesian inference, including the latent process and hyper-parameter sets, and can easily be applied to data with misalignment. We propose a block update MCMC sampling algorithm that has faster convergence in high-dimensional parameter spaces. 
%\subsubsection{Bayesian Simplified Linear Model of Coregionalization (BSLMC)}\label{sec: SLMC_model_def}
%Consider the spatial regression model with latent process modeled by SLMC.
Elements in $\fb(\bs)$ are %independent and each follows a zero centered GP with covariance function $\brho_{\psi_k}(\cdot, \cdot)$ defined through $\psi_k$ for $ k = 1, \ldots, K$. The
as described in (\ref{eq: SLMC}), while the noise process $\beps(\bs) \overset{iid}{\sim} \mathrm{N}(\mathbf{0}, \bSigma)$ with covariance matrix $\bSigma$. We model $\{\bbeta, \bLambda, \bSigma\}$ using a Matrix-Normal-Inverse-Wishart family. To be precise,  
\begin{equation}\label{eq: LMC_priors}
    \begin{aligned}
    \bbeta \given \bSigma &\sim \mbox{MN}(\bmu_{\bbeta}, \bV_{\bbeta}, \bSigma)\; ; \;
    \bLambda \given \bSigma \sim \mbox{MN}(\bmu_{\bLambda}, \bV_{\bLambda}, \bSigma)\; ; \; \bSigma \sim \mbox{IW}(\bPsi, \nu)
    %\mathrm{vec}(\boldsymbol{\omega}\mathbf{Z}) \given \mathbf{Z}, \{\phi_k\}_{k = 1}^r\sim \mathrm{N}(0, \mathrm{block}\{\sum_{k = 1}^r Z_{ki}Z_{kj}\rho_{ \phi_k}(S, S)\}_{i = 1, j = 1}^q)
    \end{aligned}\;,
\end{equation}
where $\bmu_{\bLambda}$ a $q \times K$ matrix and $\bV_{\bLambda}$ a $K \times K$ positive definite matrix. A random matrix $\mathbf{Z}_{n \times p}\sim \mbox{MN}_{n,p}(\mathbf{M}, \bU, \bV)$ has the probability density function \citep{dawid1981}
\begin{equation}\label{eq: MN_density}
p(\bZ\mid\bM, \bU, \bV) = \frac{\exp\left[ -\frac{1}{2} \, \mbox{tr}\left\{ \bV^{-1} (\bZ - \bM)^{T} \bU^{-1} (\bZ - \bM) \right\} \right]}{(2\pi)^{np/2} |\bV|^{n/2} |\bU|^{p/2}}\; , 
\end{equation}
where $\mbox{tr}(\cdot)$ is the trace function, $\bM$ is the mean matrix, $\bU$ is the first scale matrix with dimension $n \times n$ and $\bV$ is the second scale matrix with dimension $p \times p$. This distribution is equivalent to
$
  \mbox{vec}(\bZ) \sim \mbox{N}_{np}(\mbox{vec}(\bM), \bV \otimes \bU)\;,
$
where $\otimes$ is the Kronecker product and $\mbox{vec}(\bZ) = \left[\bz_1^\top, \ldots, \bz_p^\top\right]^\top$ is the vectorized $n \times p$ random matrix $\bZ = [\bz_1: \cdots : \bz_p]$. 
%We model the matrix $\left\{\left[\bbeta^\top, \bLambda^\top\right]^\top, \bSigma\right\}$  
{We refer to the model specified through \eqref{eq: SLMC_model_element}--\eqref{eq: LMC_priors} as the Bayesian LMC (BLMC) factor model.}

Without misalignment, the observation model in \eqref{eq: SLMC_model_element} can be cast as
\begin{equation}\label{eq: LMC_model}
    \bY_{n \times q} = \bX_{n \times p}\bbeta_{p \times q} + \bF_{n \times K} \bLambda_{K \times q} + \beps_{n \times q}\;,
\end{equation}
where $\bY = \by(\mathcal{S}) = [\by(\bs_1): \cdots : \by(\bs_n)]^\top$ is the $n \times q$ response matrix, $\bX = \bx(\mathcal{S}) = [\bx(\bs_1) : \cdots : \bx(\bs_n)]^\top$ is the corresponding design matrix with full rank ($n > p$), and %$\bF = [\fb_1(\calS): \cdots: \fb_K(\calS)]$ 
$\bF$ is the $n\times K$ matrix with $j$-th column being the $n\times 1$ vector comprising $f_j(\bs_i)$'s for $i=1,2,\ldots,n$.

{The parameters $\bLambda$ and $\bF$ are not jointly identified in factor models and some constraints are required to ensure identifiability \citep{Lopes2004,renbanerjee2013}. These constraints are not without problems. For example, a lower-trapezoidal (triangular for $K=q$) specification for $\bLambda$ imposes possibly unjustifiable conditional independence on the spatial processes. Alternatively, ordering the spatial range parameters can ensure identifiability but creates difficulties in computation and interpretation. We avoid such constraints and transform $\bomega = \bF\bLambda = [\bomega(\bs_1) : \cdots : \bomega(\bs_n)]^\top$ to obtain inference for the latent process. This parametrization yields conditional conjugate distributions and, therefore, efficient posterior sampling. We elucidate below in the context of misaligned data. %To evaluate the associations among the latent processes, we compute $p(\bOmega \given \{\by(\bs_i)_{os_i}\}_{i = 1}^n)$ where $\bOmega = \frac{1}{n} \sum_{i = 1}^n(\bomega(\bs_i) - \bar{\bomega})(\bomega(\bs_i) - \bar{\bomega})^\top$ with $\bar{\bomega}$ the vector of column means of $\bomega$ is the finite sample covariance among the latent processes. 
}

{\subsection{Inference for spatially misaligned data}\label{subsec: Bayesian_SLMC_Misaligned}}
%The assigned priors in (\ref{eq: LMC_priors}) %fail to provide a conjugate prior for the SLMC model with a random $\fb(\bs)$. Whereas the design of the prior 
%yield conditional posterior distributions in a closed form for all the parameters, except $\{\psi_k\}_{k = 1}^K$. This supports a block update MCMC algorithm for posterior sampling (see Section~\ref{subsec: Block_MCMC}).
%\paragraph{Conditional posterior distributions in SLMC}
{Let $\calS = \{\bs_1, \ldots, \bs_n\}$ be the set of locations that have recorded at least one of the observed outcomes and let $\calS_i$ be the subset of locations that have recorded the $i$-th response. Then $\cup_{i = 1}^q\calS_i = \calS$ and let $n_i = |\calS_i|$.  Let $\calM_i = \calS \setminus \calS_i$ denote the set of locations where at least one response, but not the $i$th response, is recorded so that $\calM = \cup_{i = 1}^q\calM_i$ is the set of all locations with incomplete data.
} 
We derive the conditional distribution of $\bF$ and of the unobserved responses $\{y_i(\calM_i)\}_{i = 1}^q$ conditional on $\{\bbeta, \bLambda, \bSigma, \{\psi_k\}_{k = 1}^K\}$. Let %$\bI_{\calS_i, \cdot}$ is constructed by keeping the rows indexed by $\calS_i$ of an $n \times n$ identity matrix, 
%$\bbeta_i$ is the regression coefficient vector for $i$-th response, 
$\mathbf{P}$ be the %$nq \times nq$ \textcolor{orange}{From Lu: P is not $nq \times nq$ when having misalignment}
{ $(\sum_{i = 1}^q n_i) \times nq$ matrix such that $\mathbf{P}\mbox{vec}(\bY) = [\by(\bs_1)^{\T}_{os_1}, \by(\bs_2)^{\T}_{os_2},\ldots,\by(\bs_n)^{\T}_{os_n}]^{\top}$},
%\begin{equation}\label{eq: permute_matrix}
%    \mathbf{P}\left[ \begin{array}{c} y_1(\calS)\\ \vdots \\ y_q(\calS) \end{array} \right] = \left[ \begin{array}{c} \by(\bs_1)_{os_1}\\ \vdots \\ \by(\bs_n)_{os_n} \end{array} \right]
%    \; , 
%\end{equation}
where the suffix $os_i$ indexes of the observed responses at $\bs_i \in \calS$. Thus, $\mathbf{P}$ extracts the observed responses from $\mbox{vec}(\bY)$ in each of the locations $\{\bs_1, \ldots, \bs_n\}$. The joint distribution of $\mathrm{vec}(\bF)$ and $\{\by(\bs_i)_{os_i}\}_{i = 1}^n$, given $\{\bbeta, \bLambda, \bSigma, \{\psi_k\}_{k = 1}^K\}$, can be represented through the augmented linear system, 
\begin{equation}
\begin{array}{c}
\left[ \begin{array}{c} \{(\by(\bs_i) - \bx(\bs_i)^\top \bbeta)_{os_i}\}_{i = 1}^n\\ \bzero \end{array} \right]
=
\left[ \begin{array}{c} \mathbf{P} (\bLambda^\top \otimes \bI_n)\\ \bI_K \otimes \bI_n \end{array} \right] \mbox{vec}(\bF)
%\left[ \begin{array}{c} f_1(\calS) \\ \vdots\\f_K(\calS) \end{array} \right] 
+ \left[ \begin{array}{c}  \beps_1 \\ \beps_2 \end{array} \right] \, ,
\end{array}
\end{equation}
where $\beps_1 \sim \mathrm{N}(\mathbf{0},\, \oplus_{i=1}^n \{\bSigma_{os_i}\})$, $\beps_2 \sim \mathrm{N}(\mathbf{0},\, \oplus_{k=1}^K \{\brho_{\psi_k}(\calS, \calS)\})$, $\brho_{\psi_k}(\calS,\calS)$ is the $n\times n$ spatial correlation matrix corresponding to $\fb_k = (f_k(\bs_1),f_k(\bs_2),\ldots,f_k(\bs_n))^{\top}$, and $\oplus_{i=1}^n$ represents the block diagonal operator stacking matrices along the diagonal. Letting $\bD_{\bSigma_o}^{-\frac{1}{2}} = \oplus_{i=1}^n\{\bSigma_{os_i}^{-\frac{1}{2}}\}$ and $\bV_{\bF} = \oplus_{k=1}^K\{\bV_k\}$, where $\brho_{\psi_k}^{-1}(\calS, \calS) = \bV_k^\T\bV_k$, we obtain
\begin{equation}\label{eq: argument_liner_LMC}
\begin{array}{c}
\underbrace{\left[ \begin{array}{c} \bD_{\bSigma_o}^{-\frac{1}{2}}\{(\by(\bs_i) - \bx(\bs_i)^\top \bbeta)_{os_i}\}_{i = 1}^n\\ \bzero \end{array} \right]}_{\Tilde{\bY}}
=
\underbrace{\left[ \begin{array}{c} \bD_{\bSigma_o}^{-\frac{1}{2}}\mathbf{P}\bLambda^\top \otimes \bI_n\\ \bV_{\bF} \end{array} \right] }_{\Tilde{\bX}}
\mbox{vec}(\bF)
%\underbrace{\left[ \begin{array}{c} f_1(\calS) \\ \vdots\\ f_K(\calS) \end{array} \right]}_{\mathrm{vec}(\bF)} 
+ \underbrace{\left[ \begin{array}{c}  \etab_1 \\ \etab_2 \end{array} \right]}_{\Tilde{\etab}}
\end{array}\, .
\end{equation}
The elements of $\Tilde{\etab}$ are independent error terms, each with unit variance. The full conditional distribution $\mbox{vec}(\bF) \given \{\by(\bs_i)_{os_i}\}_{i = 1}^n, \bbeta, \bLambda, \bSigma, \{\psi_k\}_{k = 1}^K$ for the LMC model in (\ref{eq: SLMC_model_element}) then follows
\begin{equation}\label{eq: SLMC_F_cond_post}
    \mbox{vec}(\bF) \given \{\by(\bs_i)_{os_i}\}_{i = 1}^n, \bbeta, \bLambda, \bSigma, \{\psi_k\}_{k = 1}^K \sim \mathrm{N}((\Tilde{\bX}^\T\Tilde{\bX})^{-1}\Tilde{\bX}^\T\Tilde{\bY}, \, (\Tilde{\bX}^\T\Tilde{\bX})^{-1}).
\end{equation}

{For misaligned data, we will perform Bayesian updating of the outcomes missing at a location $\bs\in {\cal M}$. Let $ms$ be the suffix that indexes outcomes that are missing at $\bs \in \calM$.} 
The conditional distribution of $\by(\bs)_{ms}$ given the parameters $\left\{\bF, \{\by(\bs_i)_{os_i}\}_{i = 1}^n, \bbeta, \bLambda, \bSigma\right\}$ is
\begin{equation}\label{eq: SLMC_YM_cond_post}
\mbox{N}([\bmu_{\bs}]_{ms} + \bSigma_{[ms, os]}\bSigma_{[os, os]}^{-1}(\by(\bs)_{os} - [\bmu_{\bs}]_{os}), \bSigma_{[ms, ms]} - \bSigma_{[ms, os]}\bSigma_{[os, os]}^{-1}\bSigma_{[os, ms]})\; ,
\end{equation}
where $\bmu_{\bs} = \bbeta^\top \bx(\bs) + \bLambda^\top \fb(\bs)$, $\bSigma_{[ms, os]}$ is the sub-matrix of $\bSigma$ extracted with row and column indices $ms$ and $os$, respectively%, and $\bSigma_{[os, os]}^{-1}$ is the inverse of $\bSigma_{[os, os]}$
. %Next, assume $\bF$ and the unobserved response $\{y_i(\calM_i)\}_{i = 1}^q$ over $\calS$ are given. 
With the priors given in \eqref{eq: LMC_priors}, we let $\bV_\Lambda = \bL_\Lambda\bL_\Lambda^\T$ and define $\bgamma = [\bbeta^\top, \bLambda^\top]^\top$. The conditional posterior distribution $\bgamma \given \bSigma, \bF, \bY$ can be found from
\begin{equation}\label{eq: augment_linear_LMC}
\begin{array}{c}
\underbrace{ \left[ \begin{array}{c} \bY\\ \bL_{\bbeta}^{-1} \bmu_{\bbeta} \\ \bL_{\bLambda}^{-1} \bmu_{\bLambda} \end{array} \right]}_{\bY^\ast}
= \underbrace{ \left[ \begin{array}{cc} \bX& \bF \\ \bL_{\bbeta}^{-1}& \mathbf{0} \\  \mathbf{0}& \bL_{\bLambda}^{-1} \end{array} \right] }_{\bX^{\ast}}  \underbrace{\left[ \begin{array}{c} \bbeta \\ \bLambda \end{array} \right]}_{\bgamma}+
\underbrace{ \left[ \begin{array}{c} \etab_1 \\ \etab_2 \\ \etab_3 \end{array} \right]}_{\etab^\ast}\; ,
\end{array}
\end{equation}
where $\etab^\ast \sim \mbox{MN}(\mathbf{0}_{(n + p + K) \times q}, \bI_{n + p + K}, \bSigma)$. Using standard distribution theory, we can show that $\bgamma, \bSigma \given \bF, \bY$ follows $\mbox{MNIW}(\bmu^\ast, \bV^\ast, \bPsi^\ast, \nu^\ast)$, where
%{\small
\begin{equation}\label{eq: SLMC_MNIW}
\begin{aligned}
    \bV^\ast = [\bX^{\ast\top}\bX^{\ast}]^{-1}, \; \bmu^\ast = \bV^\ast[\bX^{\ast\top}\bY^{\ast}],\;
    \bPsi^\ast = \bPsi + \bS^{\ast}
%    (\bY^{\ast} - \bX^{\ast}\bmu^\ast)^{\top}(\bY^{\ast} - \bX^{\ast}\bmu^\ast)
,\;\mbox{ and }\;
    \nu^\ast = \nu + n
\end{aligned}
\end{equation}
%}
with $\bS^{\ast} = (\bY^{\ast} - \bX^{\ast}\bmu^\ast)^{\top}(\bY^{\ast} - \bX^{\ast}\bmu^\ast)$.
In particular, if $\bSigma = \oplus_{i=1}^q\{\sigma^2_i\}$ {and each $\sigma^2_i \sim \mbox{IG}(a, b_{i})$} for $i = 1, \ldots q$, then the conditional distribution of $\sigma^2_i$ given $\bY, \bF$ follows $\mbox{IG}(a^\ast, b_i^\ast)$, where 
\begin{equation}\label{eq: Diag_Sigma_LMC_post}
    a^\ast = a + \frac{n}{2}\; , \; b_i^\ast = b_i + \frac{1}{2} (\bY^\ast - \bX^\ast\bmu^\ast)_i^\top (\bY^\ast - \bX^\ast\bmu^\ast)_i\;, \; i = 1, \ldots, q\;,
\end{equation}
and $(\bY^\ast - \bX^\ast\bmu^\ast)_i$ is the $i$-th column of $\bY^\ast - \bX^\ast\bmu^\ast$. From \eqref{eq: augment_linear_LMC}, $ %\left. \left[ \begin{array}{c} \bbeta \\ \bLambda \end{array} \right]\, \right \vert \,
\bgamma \given 
\bSigma, \bF, \bY \sim \mbox{MN}(\bmu^\ast, \bV^\ast, \bSigma)$. 
%A diagonalized $\bSigma$ depicts the variance of the measurement error for each response. %Diagonalization of $\bSigma$ also allows the SLMC model to serve as a factor model when the number of response $q$ is large. 

{The parameters $\psi_{k}$, $k=1,2,\ldots,K$, by themselves, are not consistently estimable under in-fill asymptotics. Therefore, irrespective of the sample size (within a fixed domain), inference on $\psi_{k}$ will be sensitive to the choice of the prior. Furthermore, without placing restrictions on the loading matrix or ordering these parameters \citep{renbanerjee2013}, these parameters are identifiable primarily through the prior. We treat these as unknown and model them using priors based upon customary spatial domain considerations.} The full conditional distributions for $\{\psi_k\}_{k = 1}^K$ are not available in closed form. However, since $\{\psi_k\}_{k = 1}^K$ and $\bY$ are conditionally independent given $\{\bF, \bgamma, \bSigma\}$, and $\fb_k$ are independent for $k=1,2,\ldots,K$, we obtain $p(\psi_k \given \bF, \bY, \bgamma, \bSigma, \{\psi_j\}_{j \neq k})$ up to a proportionality constant as
\begin{equation}\label{eq: SLMC_psi_cond_post}
 p(\bY \given \bF, \bgamma, \bSigma)\times p(\bgamma, \bSigma) \times \prod_{k = 1}^Kp(\fb_k \given \psi_k) \times p(\psi_k) \propto p(\fb_k \given \psi_k) \times p(\psi_k)\;,
\end{equation}
for each $k = 1,\ldots,K$, where $p(\psi_k)$ is the prior for $\psi_k$. %Often, the right hand side of \eqref{eq: SLMC_psi_cond_post} is much easier to calculate than a direct formulation of the posterior distribution of $\psi_k$, especially when $\fb_k$'s are modeled using scalable spatial models.  %Thus facilitates the process of obtaining posterior inference for $\{\psi_k\}_{k = 1}^K$.
%\textcolor{blue}{Scalable processes for the $\fb_k$'s yield significant computational gains.}  

Turning to predictions, if $\calU = \{\bu_1, \ldots, \bu_{n'}\}$ is a set of new locations, then $\bY_\calU = \by(\calU)$ is independent of $\{\by(\bs_i)_{os_i}\}_{i = 1}^n$ given $\{\bbeta, \bLambda, \bSigma\}$ and $\bF_{\calU} = [\fb_1(\calU): \ldots : \fb_K(\calU)]^\top$. Then,
\begin{equation}\label{eq: FU_LMC_post}
  \fb_k(\calU) \given \fb_k, \psi_k \sim \mbox{N}(\brho_{\psi_k}(\calU, \calS) \brho^{-1}_{\psi_k}(\calS, \calS)\fb_k, \, \brho_{\psi_k}(\calU, \calS) \brho^{-1}_{\psi_k}(\calS, \calS)\brho_{\psi_k}(\calS, \calU)) \; ,
\end{equation}
for each $k=1,2,\ldots,K$. %\textcolor{orange}{From Lu: Maybe we can discuss here, this step is often scalable when $\fb_k$ are modeled using scalable spatial models.}This step can be can render computational benefits using scalable spatial processes for $\fb_k$. 
It follows that $p(\bY_\calU, \bF_\calU \given  \{\by(\bs_i)_{oi}\}_{i = 1}^n)$ is proportional to
\begin{equation}\label{eq: post_pred_SLMC}
\begin{aligned}
    p(\bY_\calU \given  \bF_\calU, \bbeta, \bLambda, \bSigma) \times p(\bF_\calU \given \bF, \{\psi_k\}_{k = 1}^K) \times p(\bbeta, \bLambda, \bSigma , \bF, \{\psi_k\}_{k = 1}^K  \given \{\by(\bs_i)_{os_i}\}_{i = 1}^n)\; ,
\end{aligned}
\end{equation}
where we have used the independence between $\bF_\calU$ and $\{\by(\bs_i)_{os_i}\}_{i = 1}^n$ given $\bF$ and $\{\psi_k\}_{k = 1}^K$. The distributions in \eqref{eq: FU_LMC_post} and \eqref{eq: post_pred_SLMC} help in sampling from the posterior predictive distribution over $\calU$ using the posterior samples of $\left\{\bbeta, \bLambda, \bSigma, \bF, \{\psi_k\}_{k = 1}^K\right\}$. We elaborate below.

\subsection{The block update MCMC algorithm}\label{subsec: Block_MCMC}
We formulate an efficient MCMC algorithm for obtaining full Bayesian inference as follows. From the $l$th iteration with $\{\bbeta^{(l)}, \bLambda^{(l)}, \bSigma^{(l)}, \{\psi_k^{(l)}\}_{k = 1}^K\}$, we generate $\bF^{(l + 1)}$ from (\ref{eq: SLMC_F_cond_post}). {Next, we draw $\{\by(\bs_i)_{mi}^{(l + 1)}\}_{\bs_i\in \calM}$ on $\calM$ using (\ref{eq: SLMC_YM_cond_post})} and then update $\{\bbeta^{(l + 1)}, \bLambda^{(l + 1)}, \bSigma^{(l + 1)}\}$ using (\ref{eq: SLMC_MNIW}). We complete the $(l+1)$th iteration by drawing $\{\psi_k^{(l + 1)}\}_{k = 1}^{K}$ through a Metropolis random walk using (\ref{eq: SLMC_psi_cond_post}). Upon convergence, these iterations will generate samples from the desired joint posterior distribution $p(\bbeta,\bLambda, \bF, \bSigma,\{\psi_k\}_{k=1}^K\given \bY)$.  

For inference on $\calU$, we sample $\bF_\calU$ from \eqref{eq: FU_LMC_post}, given the posterior samples of $\bF$ and $\{\psi_k\}_{k = 1}^K$, then generate posterior predictions of $\bY_\calU$ given the posterior samples of $\{\bbeta, \bLambda, \bSigma, \bF_\calU\}$.
\iffalse
\begin{enumerate}
    \item  $\{\by(\bs_i)_{mi}\}_{\bs_i\in \calM}$ over $\calM$ through \eqref{eq: SLMC_YM_cond_post}. 
    \item Sample $\{\bbeta, \bLambda, \bSigma\}$ from their full conditional posterior distribution \eqref{eq: SLMC_MNIW}
    \item Use \eqref{eq: SLMC_psi_cond_post} to update the MCMC chains for each element of $\{\psi_k\}_{k = 1}^{K}$ through Metropolis-Hasting (M-H) algorithm.
    \item Sample posterior predictions 
    of $\bF_\calU$ through \eqref{eq: FU_LMC_post} given posterior sample of $\bF$ and $\{\psi_k\}_{k = 1}^K$, then generate posterior predictions of $\bY_\calU$ given posterior sample of $\{\bbeta, \bLambda, \bSigma, \bF_\calU\}$ based on the assumptions of SLMC model.
\end{enumerate}
\fi
%After convergence, we can obtain the full Bayesian inference through the MCMC chains. 
%The random walk proposal is only involved in the sampling process of the hyperparameter set $\{\psi_k\}_{k = 1}^{K}$, the rest of the parameters are directly sampled through the posterior conditional distribution. 
%The posterior prediction $f_k(\calU)$ can be generated using \eqref{eq: FU_LMC_post} given the posterior samples $\bF, \psi_k$ for $k = 1, \ldots, K$. Then the posterior prediction $\bY_\calU$ can be obtained through the definition of SLMC model \eqref{eq: SLMC_model_element} with the posterior samples of $\{\bbeta, \bLambda, \bSigma\}$ and $\bF_\calU$. 
%It can be shown that the resulting MCMC chains converge to the posterior distribution, using standard Markov chain theory \citep{tierney1994markov}. 
Applying the SCAM algorithm introduced in \citet{haario2005componentwise}, one can avoid tuning parameters in Metropolis algorithm by warming up each MCMC chain of $\{\psi_k\}_{k = 1}^K$ with an adaptive proposal distribution. 
%Since we update a single parameter each time and $\{\psi_k\}_{k = 1}^K$ have low correlation, any adaptive MCMC algorithm should quickly provide a reasonable proposal distribution. 
In our implementation, we use the proposal distribution defined by equation (2.1) in \citet{roberts2009examples}%with an empirical estimate of the covariance of the target distribution based on half of the chain's history
. 
%We ran the adaptive algorithm for the first quarter of the MCMC chains and then fixed the proposal distribution for the rest of the MCMC chains in the simulation studies and real data analysis See Sections~\ref{sec: simulation}~and~\ref{sec: real_data_analy}.

%The parameters $\bLambda$ and $\bF$ are not jointly identified, but we can transform back to $\bomega = \bF\bLambda$ and obtain inference for the latent process. This parametrization has the advantage of conditional conjugacy, which brings more efficient computation in posterior sampling. 
We sample $\bF$ as a single block through a linear transformation of the $n \times K$ independent parameters from the model in \eqref{eq: argument_liner_LMC}. Sampling $\{\bbeta, \bLambda\}$ follows analogously. We significantly improve convergence by reducing the posterior dependence among the parameter in this Gibbs with Metropolis algorithm \citep{gelman2013}. %Though the dimension of the posterior space is linear to the number of observations, 
%We expect a top convergence rate comparing to other Bayesian spatial models with MCMC algorithm sampling over such a high dimension posterior space. %that is comparable to a Gibbs sampling algorithm with a low dimension of parameters. 
Since $\bF$ is sensitive to the value of the intercept, we recommend using an intercept-centered latent process to obtain inference for the latent spatial pattern. %We also advise against initializing $\bLambda = \bO$ in the MCMC, which may result in $\bF$ getting an extreme initial value and slowing down convergence. 

\subsection{Scalable Modeling}\label{sec: model_implement}
%Analogous to the conjugate multivariate latent model \textcolor{blue}{\citet{zhang2020high}}, 
We use a conjugate gradient method \citep{nishimura2018prior} to facilitate sampling of $\bF$ when $\brho_{\psi_k}^{-1}(\calS, \calS)$ is sparse for $k = 1, \ldots, K$. %Accelerating MCMC sampling through a conjugate gradient method has an elaborate implementation in \citet{nishimura2018prior}. %We develop a Bayesian framework to implement this sampling scheme in massive multivariate spatial data modeling. 
Here, we develop a scalable BLMC model with each element of $\fb(\bs)$ modeled as a Nearest-Neighbor Gaussian Process (NNGP).

Let each $f_k(\bs), \bs \in \mathcal{D}$ be an $\mbox{NNGP}(0, \rho_{\psi_k}(\cdot, \cdot))$, which implies that $\fb_k \sim \mbox{N}(\bzero, \Tilde{\brho}_k)$ for each $k=1,2,\ldots,K$, where $\Tilde{\brho}_k = (\bI - \bA_{\rho_k})^{-1}\bD_{\rho_k}(\bI - \bA_{\rho_k})^{- \top}$, $\bA_{\rho_k}$ is a sparse-lower triangular matrix with no more than a specified small number, $m$, of nonzero entries in each row and $\bD_{\rho_k}$ is a diagonal matrix. The diagonal entries of $\bD_{\rho_k}$ and the nonzero entries of $\bA_{\rho_k}$ are obtained from the conditional variance and conditional expectations for a Gaussian process with covariance function $\rho_{\psi_k}(\bs,\bs')$. We consider a fixed order of locations in $\calS$ and let $N_m(\bs_i)$ be the set of at most $m$ neighbors of $\bs_i$ among locations $\bs_j\in \calS$ such that $j<i$. The $(i,j)$-th entry of $\bA_{\rho_k}$ is $0$ whenever $\bs_j \notin N_m(\bs_i)$. If $j_1 < j_2 < \cdots < j_m$ are the $m$ column indices for the nonzero entries in the $i$-th row of $\bA_{\rho_k}$, then the $(i,j_k)$-th element of $\bA_{\rho_k}$ is the $k$-th element of the $1\times m$ vector $\ba_i^{\top} = \brho_{\psi_k}(\bs_i, N_m(\bs_i))\brho_{\psi}(N_m(\bs_i),N_m(\bs_i))^{-1}$. The $(i,i)$-th diagonal element of $\bD_{\rho_k}$ is given by $\rho_{\psi_k}(\bs_i,\bs_i) - \ba_i^{\top}\brho_{\psi_k}(N_m(\bs_i),\bs_i)$. Repeating these calculations for each row completes the construction of $\bA_{\rho_k}$ and $\bD_{\rho_k}$ and yields a sparse $\Tilde{\brho}_k^{-1}$. This construction is performed in parallel and requires storage or computation of at most $m\times m$ matrices, where $m << n$, costing $\mathcal{O}(n)$ flops and storage. {%A detailed algorithm of the MCMC algorithm for NNGP based BLMC model is presented in Web Appendix B.
See Appendix~\ref{SM: BSLMC_NNGP_alg} for details.}

%\textbf{Lu: Please place your Algorithm 5 in the supplement here now.}

%\subsection{BSLMC Factor model with diagonal $\bSigma$} 

%\textcolor{blue}
{%Here, we briefly extend the discussion of the scalable modeling strategies for the proposed MCMC algorithm. It is crucial to see 
Sampling $\bF$ is computationally expensive, but is expedited by solving $(\tilde{\bX}^\top\tilde{\bX})^{-1}\tilde{\bX}^\top \bv$ efficiently for any vector $\bv$. If $\brho_{\psi_k}(\calS, \calS)= \bL_k\bL_k^\top $ has a sparse Cholesky factor $\bL_k$, then %with $\bL_k^{-1} = \bV_k$ 
calculating $\tilde{\bX}^\top \bv$ is efficient. To be precise, the Woodbury matrix identity yields
\begin{equation}\label{eq: woodbury}
    (\tilde{\bX}^\top\tilde{\bX})^{-1} = (\bF \bD_{\bSigma_0}^{-1}\bF^\top + \oplus_{k=1}^K \{\brho_k^{-1}\})^{-1} = \oplus_{k=1}^K \{\brho_k\} - \oplus_{k=1}^K \{\brho_k\}\bF \bG^{-1}\bF^{\top}  \oplus_{k=1}^K \{\brho_k\}\;,
\end{equation}
where $\bF = \left(\bLambda \otimes \bI_n\right) \bP^{\top}$ is sparse, $\bG = \bD_{\bSigma_0} + \bP\{\sum_{k = 1}^K \lambda_{ik}\lambda_{jk}\brho_k \}_{i,j = 1}^{p}\bP^\top$ with $\brho_k = \brho_{\psi_k}(\calS, \calS)$. If all the $\brho_k$'s have similar structures, then permuting $\{\sum_{k = 1}^K \lambda_{ik}\lambda_{jk}\brho_k \}_{i,j = 1}^{p}$ with $\bP$ in rows and columns often renders structures in $\brho_k$'s that can be exploited by BLMC for very large spatial data sets  %sparsity in the Cholesky decomposition of $\brho_k$
. 
For example, if $\brho_k$'s are banded matrices with bandwidth $b$, then $\bP\{\sum_{k = 1}^K \lambda_{ik}\lambda_{jk}\brho_k \}_{i,j = 1}^{p}\bP^\top$ is also banded with bandwidth $bq$. Moreover, $\bD_{\bSigma_0}$ is a banded matrix with bandwidth $\leq q$. Hence, adding $\bD_{\bSigma_0}$ hardly increases the computational burden in the Cholesky decomposition of $\bG$ when $q$ is small. Assembling all features of $\brho_k$, $\bF$ and $\bG$, the calculation of $(\tilde{\bX}^\top\tilde{\bX})^{-1}\bu$ for any $\bu = \tilde{\bX}^\top\bv$ is scalable when multiplying $\bu$ with \eqref{eq: woodbury}.
}

We conclude this section with a remark on the BLMC model with diagonal $\bSigma$. This specification is desirable for data sets with a massive number of responses $q$. A diagonal $\bSigma$ avoids the quadratic growth of the number of parameters in $\bSigma$ as $q$ increases. %When $K < q$, it becomes a factor model that can fit the latent process with a low-rank structure.  
We illustrate an NNGP based BLMC with diagonal $\bSigma$ in Section~\ref{subsec: sim_2}.

\section{On posterior consistency: Large-sample properties of posterior estimates}\label{sec: Theory}
We present some theoretical results for the models constructed in the previous section. Specifically, we investigate the behavior of the posterior distribution as the sample size increases and establish its convergence to an oracle distribution. Here, for establishing the results, we will assume conjugate MNIW models with no misalignment. First, we assume that $\by(\bs)$ itself is modeled as a spatial process without explicitly introducing a latent process. Let
\begin{equation}\label{eq: spatial_GP_model_proportion}
\by(\bs) \sim \mbox{GP}(\bbeta^\top \bx(\bs), \bC(\cdot, \cdot)) \text{ , } \bC(\bs, \bs') = \{\rho_{\psi}(\bs, \bs') + (\alpha^{-1} - 1)\delta_{\bs = \bs'}\}\bSigma\;,
\end{equation}
where $\rho_{\psi}(\cdot, \cdot)$ is a spatial correlation function defined through hyperparameter $\psi$, $\delta$ denotes Dirac's delta function, and $\alpha^{-1}\bSigma$ is the non-spatial covariance matrix of $\by(\bs)$. The fixed scalar $\alpha$ represents the proportion of total variability allocated to the spatial process. %If we assume no misalignment, i.e., all $q$ responses are observed on the observed location ${\cal S} = \{\bs_1, \ldots,  \bs_n\} \subset \mathcal{D}$, then $\bY = \by(\mathcal{S}) = [\by(\bs_1): \cdots : \by(\bs_n)]^\top$ is the $n \times q$ response matrix and $\bX = \bx(\mathcal{S}) = [\bx(\bs_1) : \cdots : \bx(\bs_n)]^\top$ is the corresponding design matrix with full rank ($n > p$). Letting  $\{\rho_{\psi}(\bs_1, \bs_2)\}_{\bs_1 \in \calS_1, \bs_1 \in \calS_2}$ as $\brho_{\psi}(\calS_1, \calS_2)$,
%On the parameter space $\{\bbeta, \bSigma\}$, the distribution of $\bY$ after integrating out the latent processes $\bomega = \bomega(\mathcal{S}) = [\bomega(\bs_1): \cdots : \bomega(\bs_n)]^\T$ 
This implies that $\bY \given \bbeta, \bSigma \sim \mbox{MN}_{n, q}(\bX\bbeta, \bkappa,  \bSigma)$, where $\bkappa = \brho_{\psi}(\mathcal{S}, \mathcal{S}) + (\alpha^{-1} - 1)\bI_n$. We model $\{\bbeta, \bSigma\}$ using the conjugate MNIW prior
\begin{equation}\label{eq: MNIW prior chp4}
    \bbeta \given \bSigma \sim \mbox{MN}_{p, q}(\bmu_{\bbeta} , \bV_r, \bSigma)\;, \; 
    \bSigma \sim \mbox{IW}(\bPsi, \nu)\; ,
\end{equation}
with prefixed $\{\bmu_{\bbeta}, \bV_r, \bPsi, \nu\}$. Closely following the developments in \citet{gamerman2004multivariate}, we obtain the posterior distribution of $\{\bbeta, \bSigma\}$ as $\text{MNIW}(\bmu^\ast, \bV^\ast, \mathbf{\Psi}^\ast, \nu^\ast)$, where
\begin{equation}\label{eq: collapsed_spatial_pars1}
\begin{aligned}
    \bV^\ast &= (\bX^\top \bkappa^{-1} \bX + \bV_r^{-1})^{-1}\; , \; \bmu^\ast = \bV^\ast (\bX^\top \bkappa^{-1}\bY + \bV_r^{-1}\bmu_{\bbeta})\; ,\\
    \bPsi^\ast &= \bPsi + \bY^\top \bkappa^{-1}\bY
    +\bmu_{\bbeta}^\top \bV_r^{-1} \bmu_{\bbeta} - \bmu^{\ast\top} \bV^{\ast-1} \bmu^\ast\; , \mbox{ and  }\nu^\ast = \nu + n\;.
\end{aligned}  
\end{equation}
We refer to the above model as the ``response'' model.

Next, we consider the spatial regression model with the latent process,   
\begin{equation}\label{eq: conj_latent_model}
 \by(\bs) = \bbeta^\top\bx(\bs) + \bomega(\bs) + \beps(\bs)\;, \; \bs \in \mathcal{D}\;,
\end{equation}
where $\bomega(\bs) \sim \mbox{GP}(\mathbf{0}_{q \times 1}, \rho_{\psi}(\cdot, \cdot)\bSigma)$ is a latent process and $\beps(\bs) \sim \mbox{N}(\mathbf{0}_{q \times 1}, (\alpha^{-1} - 1)\bSigma)$ is measurement error. Define $\bomega = \bomega(\calS) = [\bomega(\bs_1): \cdots : \bomega(\bs_n)]^\top$. For theoretical tractability, we restrict posterior inference on $\{\bbeta, \bomega, \bSigma\}$, assuming that the scalar $\alpha$ is fixed. Assuming that the joint distribution of $\bbeta$ and $\bSigma$ are given in (\ref{eq: LMC_priors}) and that $\bomega \given \bSigma \sim \mbox{MN}_{n \times q}(\mathbf{0}, \brho_{\psi}(\calS,\calS), \bSigma)$, the posterior distribution of $\bgamma^{\top} = [\bbeta^{\top}, \bomega^{\top}]$ is $p(\bgamma, \bSigma \given \bY) = \mbox{MNIW}(\bmu_{\bgamma}^\ast, \bV^\ast, \bPsi^\ast, \nu^\ast)$, where 
\begin{equation}\label{eq: augmented_conj_post_v1_pars}
\begin{aligned}
    \bV^{\ast} &= \left[\begin{array}{cc} \frac{\alpha}{1 - \alpha} \bX^\top \bX + \bV_r^{-1} & \frac{\alpha}{1 - \alpha}\bX^\top \\ \frac{\alpha}{1 - \alpha}\bX & \brho_{\psi}^{-1}(\calS, \calS) + \frac{\alpha}{1 - \alpha}\bI_n \end{array} \right]^{-1}, \bmu_{\bgamma}^\ast = \bV^{\ast} \left[\begin{array}{c} \frac{\alpha}{1 - \alpha}\bX^\top \bY + \bV_r^{-1}\bmu_{\bbeta} \\ \frac{\alpha}{1 - \alpha} \bY \end{array} \right],\\
    \bPsi^\ast &= \bPsi + \frac{\alpha}{1 - \alpha} \bY^\top\bY + \bmu_{\bbeta}^\top \bV_r^{-1}\bmu_{\bbeta} - \bmu_{\bgamma}^{\ast\top}\bV^{\ast-1} \bmu_{\bgamma}^\ast \; \mbox{  and  } \nu^\ast = \nu + n\;,
\end{aligned}
\end{equation}
We refer to the above model as the ``latent'' model.

We establish the posterior consistency of $\{\bbeta,\bSigma\}$ for the response model (\ref{eq: spatial_GP_model_proportion}) and the latent model (\ref{eq: conj_latent_model}). For distinguishing the variables based on the number of observations, we make the dependence upon $n$ explicit. Denote $\bX(n)_{n \times p} = [\bx(\bs_1): \cdots: \bx(\bs_n)]^\top$, $\bY(n)_{n \times q} = [\by(\bs_1): \cdots: \by(\bs_n)]^\top$,  $\calS(n) = \{\bs_1, \ldots, \bs_n\}$, $\bkappa(n) = \bC(\calS(n), \calS(n)) + (\alpha^{-1} - 1)\bI_n$ and $\mathbf{J}(n) =\bX(n)^\top \bkappa(n)^{-1}\bX(n)$. %In the following results, we denote $\mathbf{P}(n) =\bX(n)^\top \bkappa(n)^{-1}\bX(n)$, $\bA \geq \mathbf{B}$ to mean that $\mathbf{A} - \mathbf{B}$ is a positive semi-definite matrix, and $\mathbf{A}_{ij}$ to be the $(i,j)$-th element of $\mathbf{A}$. %\textcolor{blue}{Since the marginal distribution of $\bSigma \given \bY$ and $\bbeta \given \bSigma, \bY$ are essentially the same for conjugate response model and conjugate latent model, proof of one model can be adapted for the other.}
Proofs and technical details are available in Appendix~\ref{SM: proofs}.

\begin{theorem}{~\rm [Theorem~S.1, Theorem S.2]}\label{thm1}
Parameter set $\{\bbeta, \bSigma\}$ is posterior consistent for both conjugate response and latent models if and only if
$\mbox{lim}_{n \to \infty}\lambda_{\min}\{\mathbf{J}(n)\} = \infty$, where $\lambda_{\min}\{\mathbf{J}(n)\}$ is the smallest eigenvalue of $\mathbf{J}(n)$. 
\end{theorem}

When the explanatory variables share the same spatial correlation with the responses, the necessary and sufficient conditions for Theorem~\ref{thm1} hold (see Remark~S.2). When the explanatory variables are themselves regarded as independent observations, the necessary and sufficient conditions in Theorem~\ref{thm1} hold (see Remark~S.3). %Remark 2.2 and 2.3 discuss two common situations for the design matrix, revealing that the condition in Theorem~\ref{thm1} is a general condition. The proof of theorems and remarks in this section is given in Web Appendix C.

\section{{Simulation}}\label{sec: simulation}
We present two simulation examples. The first compares BLMC model with other multivariate Bayesian spatial models. The second assesses our BLMC model when $K$ is not excessively large. BLMC models were implemented in Julia 1.2.0 \citep{bezanson2017julia}. We modeled the univariate processes in the proposed BLMC by NNGP. We took the Bayesian LMC model proposed by \citet{schmidt2003bayesian} as a benchmark in the first simulation example. The benchmark model was implemented in R 3.4.4 through function %\textit{spMvLM} and
\textit{spMisalignLM} in the R package \textit{spBayes} \citep{finley2007spbayes}. %, which, to the best of our knowledge, is the only R package that has a Bayesian LMC model for spatial analysis. 
%The benchmark model was run in R 3.4.4. 
%We also fitted a response NNGP model with misalignment in Julia 1.2.0 in the first example. The detailed algorithms for the response NNGP model with misalignment is in the Web Appendix D. %\ref{SM: misali_resp_NNGP}.
{The posterior inference for each model was based on MCMC chains with 5,000 iterations after a burn-in of 5,000 iterations.}
{All models were run on} a single 8 Intel Core i7-7700K CPU @ 4.20GHz processor with 32 Gbytes of random-access memory running Ubuntu 18.04.2 LTS. Convergence diagnostics and other posterior summaries were implemented within the Julia statistical environment. Model comparisons were based on parameter estimates (posterior mean and 95\% credible interval), root mean squared prediction error (RMSPE%= $n^{-1}\sum_{i = 1}^n((y_j(\bs_i) - \hat{y_j}(\bs_i))^2)^{\frac{1}{2}}, j = 1,\ldots, q$
), mean squared error of intercept-centered latent processes (MSEL% = $n^{-1}\sum_{i = 1}^n((\omega_j(\bs_i) - \hat{\omega_j}(\bs_i))^2)^{\frac{1}{2}}, j = 1, \ldots, q$)
), prediction interval coverage (CVG; the percent of intervals containing the true value), interval coverage for intercept-centered latent process of observed response (CVGL), average continuous rank probability score (CRPS; see \citet{gneiting2007strictly}) for responses, and the average interval score (INT; see \citet{gneiting2007strictly}%, %prediction interval covarage (CVG%; the percent of 95\% confidence intervals containing the true value
) for responses and run time. { We assessed convergence of MCMC chains by visually monitoring auto-correlations and checking the accuracy of parameter estimates using effective sample size (ESS) \citep[][Sec. 10.5]{gelman2013} and Monte Carlo standard errors (MCSE) with batch size 50 \citep{flegal2008markov}. %and the split $\hat{R}$ method for detecting non-stationarity \citep[][sec. 15.3.2]{stan2016stan}.
%, which is given as $\mbox{MCSE}(\theta) \equiv \{V(\theta)/\mbox{ESS}(\theta)\}^{1/2}$, where $V(\theta)$ is the variance of $\theta$ of the posterior samples, $\mbox{ESS}(\theta)$ is the effective sample size (ESS) \citep[][Sec. 10.5]{gelman2013}. 
%The effective sample size and $\hat{R}$ are calculated through Julia package "MCMCDiagnostics" \citep{MCMCDiagnostics}.
} 
To calculate the CRPS and INT, we assumed that the associated predictive distribution was well approximated by a Gaussian distribution with mean centered at the predicted value and standard deviation equal to the predictive standard error. All NNGP models were specified with at most $m = 10$ nearest neighbors.
%\textcolor{blue}{?mean absolute error? (MAE - $n^{-1} \sum_{i = 1}^n |y(s_i) - \hat{y}(s_i)|$)}

\subsection{Simulation Example 1}\label{subsec: sim_1} We simulated the response $\by(\bs)$ from the LMC model in \eqref{eq: SLMC_model_element} with $q = 2, p = 2, K = 2$ over 1200 randomly generated locations over a unit square. The size of the data set was kept moderate to enable comparisons with the expensive full GP based LMC models for experiments conducted on the computing setup described earlier. The explanatory variable $\bx(\bs)$ consists of an intercept and a single predictor generated from a standard normal distribution. An exponential correlation function was used to model $\{\rho_{\psi_k}(\cdot, \cdot)\}_{k = 1}^K$, i.e.,
$
\rho_{\psi_k}(\bs, \bs') = \exp{(-\phi_k\|\bs - \bs'\|)}, \text{ for } \bs, \bs' \in {\mathcal{D}}\;,
$
where $\|\bs - \bs'\|$ is the Euclidean distance between $\bs$ and $\bs'$, and $\psi_k = \phi_k$ is the decay for each $k$. %Let $\bLambda$ in \eqref{eq: SLMC_model_element} be an upper triangular matrix. 
We randomly picked 200 locations for predicting each response to examine the predictive performance. 
%Table~\ref{table:sim2}
{Appendix~\ref{sm: values_examples} presents the fixed parameters generating the data and the subsequent posterior estimates.} %including the non-spatial covariance matrix of the $\bomega = \bF\bLambda$ (labeled as $\bOmega = \bLambda^\top \bLambda + \bSigma$). \textcolor{blue}{The posterior samples of $\mbox{cov}(\bomega)$ is generated by the non-spatial covariance of the posterior samples of latent process.}

For NNGP based BLMC model, {we assigned a flat prior for $\bbeta$, which makes $\bL_{\bbeta}^{-1}$ in \eqref{eq: augment_linear_LMC} a zero matrix. The prior for $\bLambda$ followed \eqref{eq: LMC_priors} with $\bmu_{\bLambda}$ a zero matrix and $\bV_{\bLambda}$ a diagonal matrix whose diagonal elements are 25.} The prior for $\bSigma$ was set to follow $\mbox{IW}(\bPsi, \nu)$ with $\bPsi = \mbox{diag}([1.0, 1.0])$ and $\nu = 3$. For the benchmark LMC, we assigned a flat prior for $\bbeta$, $\mbox{IW}(\bPsi, \nu)$ with $\bPsi = \mbox{diag}([1.0, 1.0])$ and $\nu = 3$ for the cross-covariance matrix $\bLambda^\top\bLambda$, and $\mbox{IG}(2, 0.5)$ for each diagonal element of $\bSigma$. 
%The candidate values for $\{\phi, \alpha\}$ were estimated using a cross-validation algorithm for the response NNGP model (with misalignment) over a 25 by 25 grid over $[2.12, 26.52] \times [0.8, 0.99]$. 
We assigned $\mbox{unif}(2.12, 212)$ as priors of decays for both models. %Given that the spatial domain is the unit square, the maximum observable distance between two points is $\sqrt{2}$. With a uniform prior between 2.12 and 212 for the spatial decay parameters 
{This
implies that the ``effective spatial range'', which is the distance where spatial correlation drops below 0.05, will be bounded above by $\sqrt{2}$ (the maximum inter-site distance within a unit square) and bounded below by 1/100th of that to ensure a wide range.} 
%The posterior inference from BLMC and the benchmark LMC model were based on an MCMC chain with \textcolor{blue}{10,000} iterations, and we took the first \textcolor{blue}{$5,000$} samples as burn-in. The number of iterations of all MCMC chains was taken to be large enough to ensure convergence. 

Table~\ref{table:sim2} presents posterior estimates of parameters and performance metrics for all candidate models. 
Both models provided similar posterior inferences for $\{\bbeta_{21}, \bbeta_{21}\}$. The 95\% credible intervals of $\{\bbeta_{11}, \bbeta_{12}\}$ all include the true value used to generate the data. 
%With a mismatch of data generating schemes and model assumptions, the response NNGP model with misalignment 
%and the conjugate latent NNGP model 
%provided incorrect inference for $\mbox{cov}(\beps)$ when compared to the other two candidate models. The RMSPEs and %the 95\% confidence interval coverage (CVG) 
%CVGs%of conjugate latent NNGP model and reponse NNGP with misalignment
%, however, are close to BLMC and benchmark LMC. Compared to benchmark LMC which cost around 21 hours, the response NNGP model spent less that 0.5 minute, suggesting that fitting the response NNGP model with misalignment is a pragmatic way to have reliable interpolation and predictions. 
The NNGP based BLMC model and the benchmark LMC model cost {2.38} minutes and around {18.25} hours, respectively. Despite the shorter running time, we observed superior performance of the NNGP based BLMC than the benchmark LMC for inferring on the latent process using CVGL{, MSEL, CRPSL and INTL}. Moreover, the interpolated map of the recovered intercept-centered latent processes (Figure~\ref{fig:sim}) %by conjugate latent NNGP,
by BLMC and benchmark LMC are almost indistinguishable from each other. BLMC and benchmark LMC produce very similar RMSPEs, CRPSs and INTs. %Benchmark LMC yields better estimates for the spatial decays. %but poorer inference for $\mbox{cov}(\bomega)$. 
The differences in estimates between the two models is likely emerging from the different prior settings and sampling schemes. Benchmark LMC restricts the loading matrix $\bLambda$ to be upper triangular, while BLMC does not, resulting in greater flexibility in fitting latent process. On the other hand, the unidentifiable parameter setting of BLMC cause less somewhat less stable inference for the hyperparameters $\{\phi_1, \phi_2\}$. {The inferences for $\{\bbeta_{11}, \bbeta_{12}\}$ are also less stable due to the sensitivity of intercept to latent process. For all other parameters including the intercept-centered latent process on 1200 locations, the median ESS is 4111.5.
%, indicating that we are obtaining an efficient Markov chain. 
All MCSEs were consistently less than 0.02. These diagnostics suggest adequate convergence of the MCMC algorithm.} %and $\hat{R}$ less than $1.05$.}

%and conjugate latent NNGP has relatively higher MSELs. 
%The simulation example shows that fitting a conjugate latent model or a response NNGP with misalignment is a pragmatic choice for quick inference in multivariate spatial data analysis. Meanwhile, BSLMC model provides more reliable inference for datasets with diverse behaviors with a longer running time, while NNGP based BSLMC is much efficient than benchmark LMC model. 

\begin{table}[!ht]
\caption{Simulation study summary table: posterior mean (2.5\%, 97.5\%) percentiles}
%\begin{adjustbox}{width=0.8\textwidth}
\begin{minipage}[t]{\textwidth} % <--- new
\centering
\scalebox{1.0}{
{
	\begin{tabular}{c|c|cc|cc}
	\hline\hline
			&  & \multicolumn{2}{c}{BLMC} & \multicolumn{2}{c}{benchmark LMC} \\
	\hline
	& true & inference & MCSE  & inference & MCSE \\
			$\bbeta_{11}$ & 1.0 & 0.705 (0.145, 1.233) & 0.034 &  0.806 (0.502,  1.131) & 0.002  \\
			$\bbeta_{12}$ & -1.0 &  -1.24 (-1.998, -0.529) & 0.045 & 
			 -1.1 (-1.533, -0.646)& 0.001 \\
			$\bbeta_{21}$ & -5.0 & -4.945 (-5.107, -4.778) & 0.002 & 
			-4.949 ( -5.113, -4.787) & 0.004 \\
			$\bbeta_{22}$ & 2.0 & 1.979 (1.78, 2.166) & 0.004 & 
			1.974 (1.785, 2.167) &  0.002  \\
			$\bSigma_{11}$ & 0.4 & 0.346 (0.283, 0.409) & 0.002 & 
			0.306 (0.248, 0.364) & 0.003 \\
			$\bSigma_{12}$ & 0.15 & 0.133 (0.072, 0.194) & 0.003 & 0.0 & -- \\
			$\bSigma_{22}$ & 0.3 & 0.29 (0.198, 0.386) & 0.004 & 
			0.233 (0.159, 0.334)& 0.005 \\
			%$\bOmega_{11}$ &  0.68 & 0.72 (0.64, 0.8) & 408 & 1.0& 0.7 (0.62, 0.78) & 68 & 1.03 \\
			%$\bOmega_{12}$ & -0.6 & -0.6 (-0.72, -0.48) & 144 & 1.0 & -0.06 (-0.11, -0.01)  & 3981 & 1.0\\
			%$\bOmega_{22}$ & 4.58 & 4.44 (4.21, 4.66) & 212 & 1.02 & 4.52 (4.33,  4.69) & 38 & 1.01\\
			$\phi_1$ & 6.0 & 8.723 (4.292, 14.065) & 0.343 & 
			12.839 (8.805, 17.471) &  0.23 \\
			$\phi_2$ & 18.0 & 22.63 (15.901, 29.555) & 0.416 & 
			18.075 (12.99, 23.741)& 0.301 \\
			\hline
			%MAE & & & & \\ 
			RMSPE$^a$
			%\footnotemark[1] 
			&-- & [0.728, 0.756, 0.742] &  & [0.725, 0.762, 0.744] & \\
			MSEL$^b$ 
			%\footnotemark[2] 
			& -- & [0.136, 0.168, 0.152] &  & [0.147, 0.192, 0.169] &  \\
			CRPS$^a$
			%\footnotemark[1] 
			& -- & [-0.412, -0.423, -0.418] &  & [-0.41, -0.427, -0.418] &  \\
			CRPSL$^b$
			%\footnotemark[2] 
			& -- & [-0.035, -0.038, -0.036] &  & [-0.216, -0.248, -0.232] & \\
			CVG$^a$
			%\footnotemark[1] 
			&-- & [0.915, 0.955, 0.935] &  & [0.925, 0.96, 0.9425] & \\
			CVGL$^b$ 
			%\footnotemark[2] 
			&-- & [0.946, 0.962, 0.954]  &  & [0.756, 0.773, 0.765] &  \\
			INT$^a$
			%\footnotemark[1] 
			&--& [3.378, 3.756, 3.567] &  & [3.347, 3.823, 3.585]& \\
			INTL$^b$
			%\footnotemark[2] 
			& --& [0.282, 0.329, 0.305] &  &  [1.875, 2.023, 1.949] & \\
			time(s) &
			%\footnotemark[3] 
			& 143 & & [42047, 23664]$^c$ &
			%\footnotemark[4]
			\\
			\hline\hline
	\end{tabular}
}
}
\footnotetext[1]{[response 1, response 2, all responses]}
\footnotetext[2]{intercept + latent process on 1000 observed locations for [response 1, response 2, all responses]} 
%\footnotetext[3]{95\% confident interval coverage rate of intercept + latent process on 1000 observed locations for [response 1, response 2, all responses]}

\footnotetext[3]{[time for MCMC sampling, time for recovering predictions]}
\end{minipage}
\label{table:sim2}
\end{table}

\begin{figure}%[h]
     \subfloat[$\bomega_1 + \bbeta_{11}$ true \label{subfig:sim2a}]{%
       \includegraphics[width=0.28\textwidth]{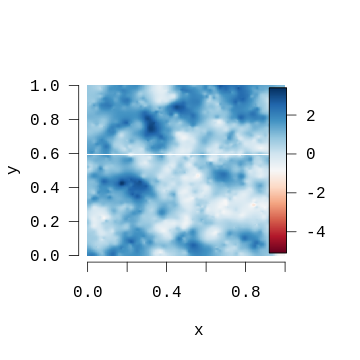}
     }
     %\hfill
     %\subfloat[$\bomega_1 + \bbeta_{11}$ conj latent \label{subfig:sim2b}]{%
     %  \includegraphics[width=0.2\textwidth]{sim2_map-w1_incp-lat.png}
     %}
     \hfill
     \subfloat[$\bomega_1 + \bbeta_{11}$ BLMC\label{subfig:sim2c}]{%
       \includegraphics[width=0.28\textwidth]{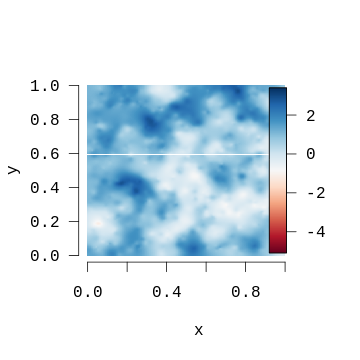}
     }
     \hfill
     \subfloat[$\bomega_1 + \bbeta_{11}$ benchmark LMC\label{subfig:sim2d}]{%
       \includegraphics[width=0.28\textwidth]{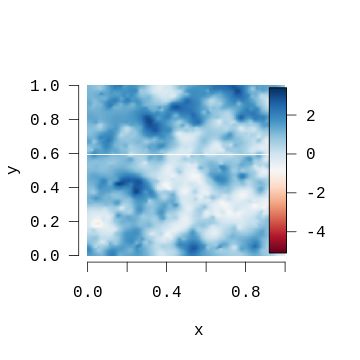}
     }\\
     
     \subfloat[$\bomega_2 + \bbeta_{12}$ true\label{subfig:sim2e}]{%
       \includegraphics[width=0.28\textwidth]{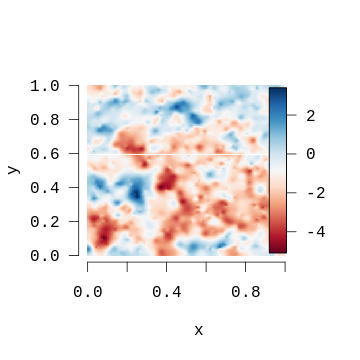}
     }
     %\hfill
     %\subfloat[$\bomega_2 + \bbeta_{12}$ conj latent\label{subfig:sim2f}]{%
     %  \includegraphics[width=0.2\textwidth]{sim2_map-w2_incp-lat.png}
     %}
     \hfill
     \subfloat[$\bomega_2 + \bbeta_{12}$ BLMC\label{subfig:sim2g}]{%
       \includegraphics[width=0.28\textwidth]{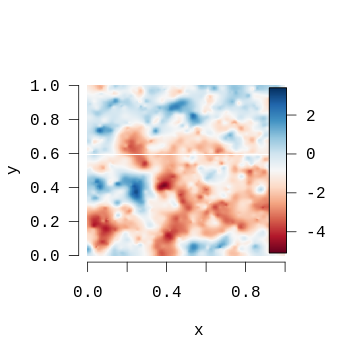}
     }
     \hfill
     \subfloat[$\bomega_2 + \bbeta_{12}$ benchmark LMC \label{subfig:sim2h}]{%
       \includegraphics[width=0.28\textwidth]{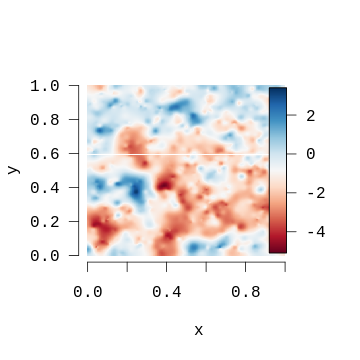}
     }
     \\
        \subfloat[fitted correlation with $K = 2$ \label{subfig:heatmap_sim3_2a}]{%
       \includegraphics[width=0.25\textwidth]{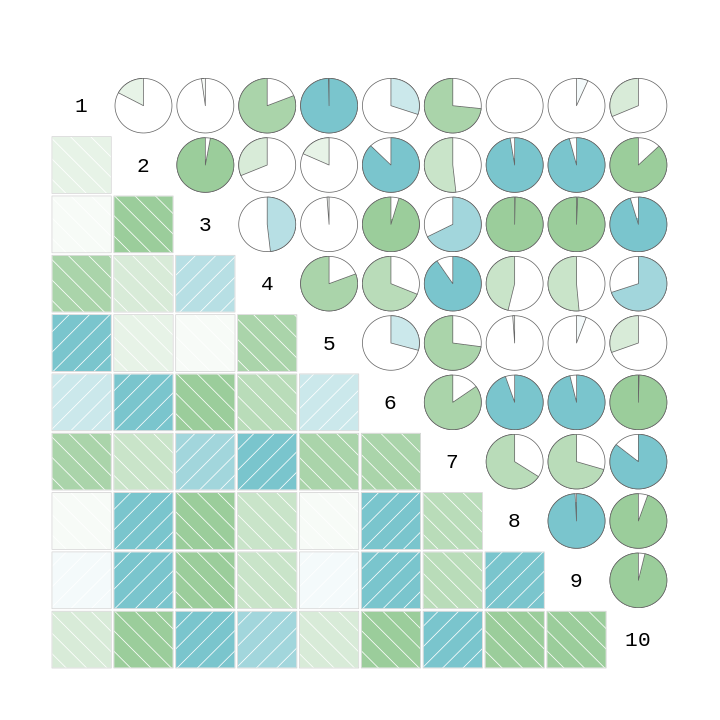}
     }
     \hfill
     \subfloat[fitted correlation with $K = 4$ \label{subfig:heatmap_sim3_2b}]{%
     \includegraphics[width=0.25\textwidth]{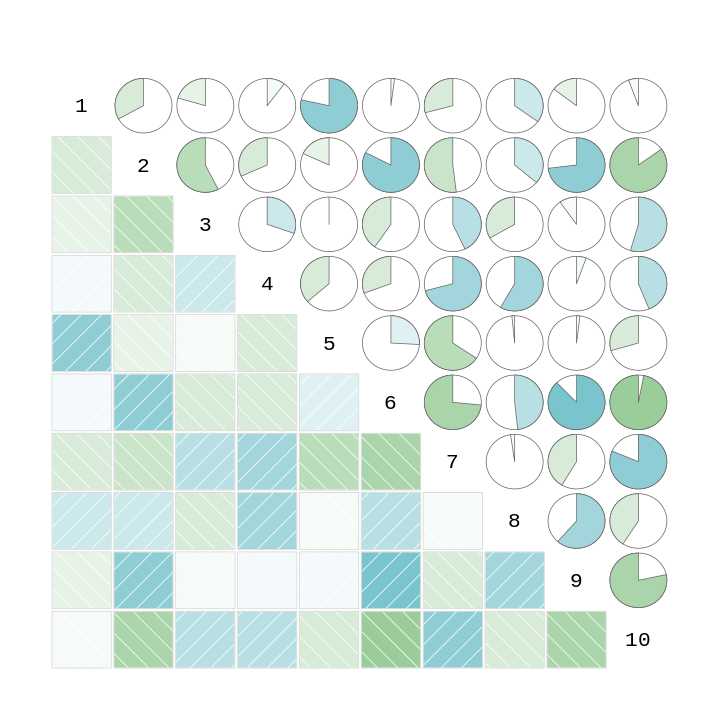}
     }
     \hfill
     \subfloat[fitted correlation with $K = 6$ \label{subfig:heatmap_sim3_2c}]{%
     \includegraphics[width=0.25\textwidth]{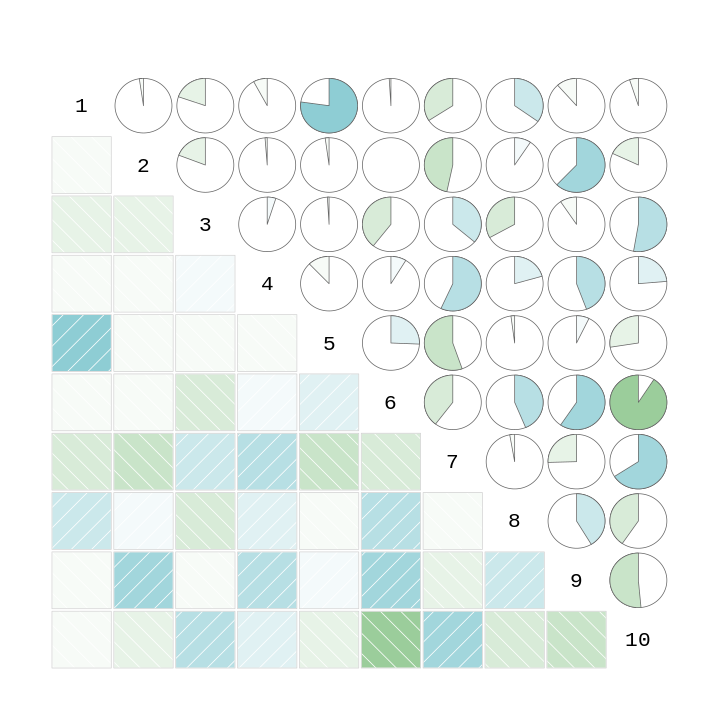}
     }\\
     
     \subfloat[fitted correlation with $K = 8$ \label{subfig:heatmap_sim3_2d}]{%
       \includegraphics[width=0.25\textwidth]{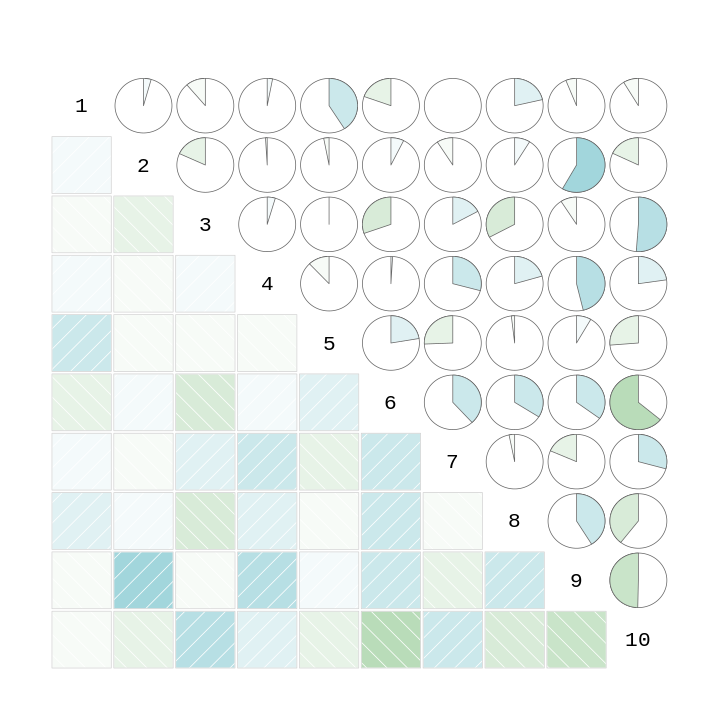}
     }
     \hfill
     \subfloat[fitted correlation $K = 10$ \label{subfig:heatmap_sim3_2e}]{%
       \includegraphics[width=0.25\textwidth]{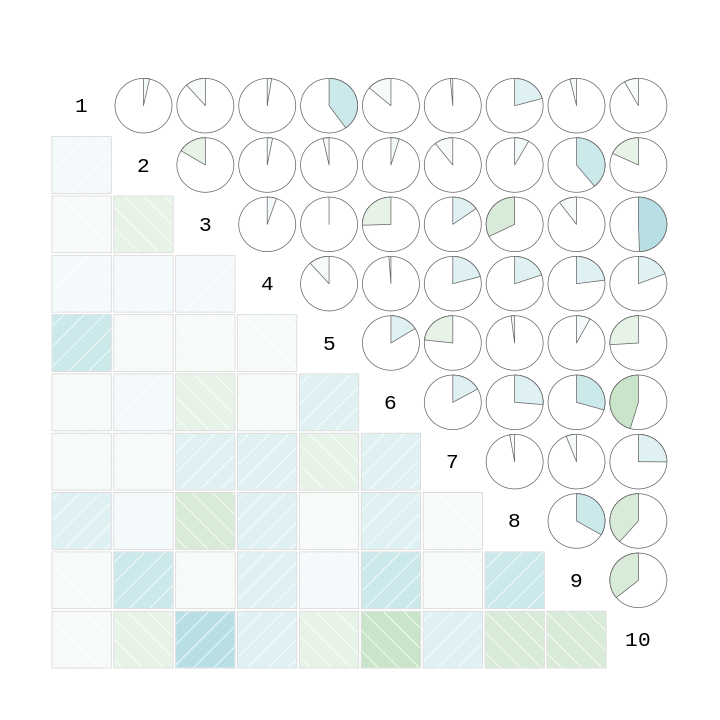}
     }
     \hfill
     \subfloat[correlation of the raw data \label{subfig:heatmap_sim3_2f}]{%
       \includegraphics[width=0.25\textwidth]{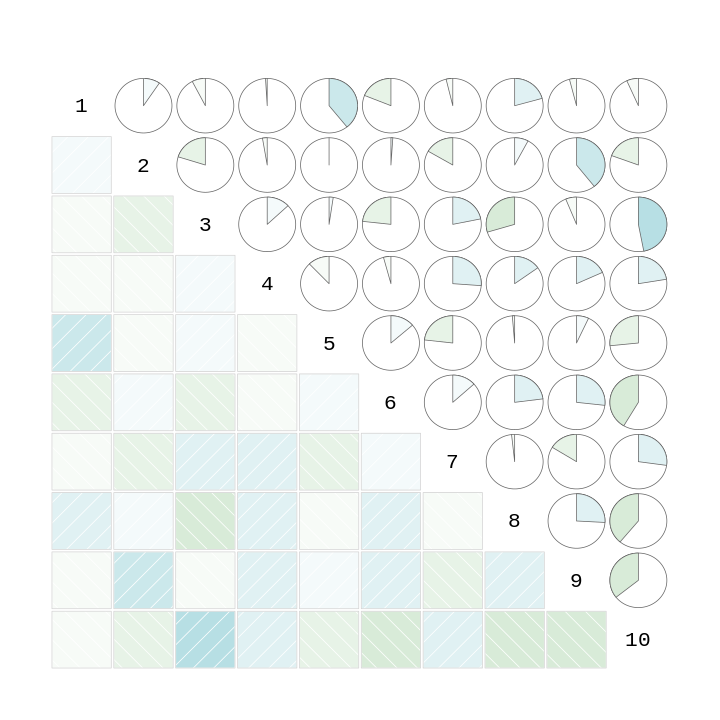}
     }
     %%%%%%%%%
     %\subfloat[actual correlation \label{subfig:heatmap_sim3_1a}]{%
     %  \includegraphics[width=0.28\textwidth]{sim3_factor_corr_plot_raw-lowK.png}
     %}
     %\hfill
     %\subfloat[fitted correlation \label{subfig:heatmap_sim3_1b}]{%
     %  \includegraphics[width=0.28\textwidth]{sim3_factor_corr_K3-lowK_plot.png}
     %}
     \caption{Interpolated maps of (a) $\&$ (d) the true generated intercept-centered latent processes, the posterior means of the intercept-centered latent process $\bomega$ from  
     the (b) $\&$ (e) NNGP based BLMC model and the (c) $\&$ (f) benchmark LMC model.
     Heat-maps of the (l) actual finite sample correlation among latent processes and (g)-(k) posterior mean of finite sample correlation among latent processes based on the posterior samples of $\bOmega$.
     \label{fig:sim}}
\end{figure}

\subsection{Simulation Example 2}\label{subsec: sim_2}
We generated {100 different} data sets using \eqref{eq: SLMC_model_element} with $\{q = 10, p = 3, K = 50\}$ and a diagonal $\bSigma$ (i.e., independent measurement errors across outcomes). {Appendix~\ref{sm: values_examples} presents the parameter values used to generate the data sets. We fixed a set of 1200 irregularly situated locations inside a unit square. The explanatory variable $\bx(\bs)$ comprised an intercept and two predictors generated independently from a standard normal distribution. The same set of locations and explanatory variables were used for the 100 data sets. Each $f_k(\bs)$ was generated using an exponential covariance function, $\{\rho_{\psi_k}(\cdot, \cdot)\}_{k = 1}^K$, where $\psi_k = \phi_k$ was the decay for $k = 1, \ldots, K$. We held out $200$ locations for assessing predictive performances.} 

%\textcolor{blue}{We repeated for 60 generated data sets with the same parameter settings in \eqref{eq: SLMC_model_element}.} These data feature a relatively large number of responses ($q = 10$) and a complicated pattern in latent processes ($K = 50$). 
{For each simulated data set}, we fitted the BLMC model specifying a diagonal $\bSigma$ with $K$ from $1$ to $10$. 
%The goal of this example is to check the performance of BLMC model, especially in recovering latent processes, when $K$ is not sufficiently large. 
Each $\phi_k$ has a Gamma prior {with shape and scale equaling 2 and 4.24, respectively, so that
%The 95\% central area under this distribution is [2.85, 65.52] and 
the %corresponding 
expected effective spatial range %, the distance beyond which spatial correlation is less than 0.05, 
is half of the maximum inter-site distance. We assigned flat prior for $\bbeta$, a vague prior for $\bLambda$ which follows the prior of $\bLambda$ in the preceding example and $\mbox{IG}(2, 1.0)$ priors for the diagonal elements of $\bSigma$.}
%The setting for the MCMC sampling scheme follows that of BLMC in the first example. 

%The running time for executing the models along with 
{
The posterior mean and the 95\% credible interval of CVGL, CVG, RMSPE and diagnostics metric %ESS and 
MCSE for regression slopes and $\bSigma$ for 100 simulation studies are summarized by $K$ in Table~\ref{table:sim3_hoff}. } 
%We also added CVG-slope in table~\ref{table:sim3_hoff}, which calculates the percent of 95\% credible intervals of regression slopes containing the true value. 
{Inference for CVG and MCSE were robust to the choice of $K$. All of the 95\% credible intervals for CVG and MCSE were within [0.9, 0.99] and [0.0, 0.02], respectively.} %, while RMSPE decreased rapidly as $K$ increased. %We also found that factor BLMC would fit the latent processes better for some of the responses than the others when $K$ was small. %The responses with low CVGL tended to have large estimation of the variance of the measurement error since the spatial effects that could not be explained by the fitted BSLMC were treated as the residuals. 
%The performance metrics quickly improved as $K$ increased from $1$ to $10$: In average, RMSPE decreased by 30.9\% and CVGL increased from 28\% to 95\%. 
%We compare the correlation across different latent processes (referred as non-spatial correlation) to check the performance of different models in estimating the latent processes. 
As shown in Table~\ref{table:sim3_hoff}, the performance metrics were quickly improved as $K$ increased from $1$ to $10$. {On average, RMSPE decreased by about 30.9\% and CVGL increased from 28\% to 95\%.} Given that our data comes from an LMC model with $K = 50$, we can conclude that BLMC with diagonal $\bSigma$ is efficient in obtaining inference for the latent processes even when $K$ is not adequately large. {We also create heat-maps of the posterior mean of our finite sample correlation matrix among the latent processes based on posterior samples of $\bOmega$, 
%To evaluate the associations among the latent processes, we compute $p(\bOmega \given \{\by(\bs_i)_{os_i}\}_{i = 1}^n)$ 
where $\bOmega = \frac{1}{n} \sum_{i = 1}^n(\bomega(\bs_i) - \bar{\bomega})(\bomega(\bs_i) - \bar{\bomega})^\top$ with $\bar{\bomega}$ the vector of column means of $\bomega$. %is the finite sample covariance among the latent processes. 
Figures~\ref{subfig:heatmap_sim3_2a}---\ref{subfig:heatmap_sim3_2e} depict such heat maps from one of the 100 simulated data sets. As $K$ increases from $2$ to $10$, the estimated correlation matrix approaches the true correlation matrix. %It can be seen that the heat-map with $K = 10$ shared a similar pattern with the finite sample correlation among the actual latent processes. 
The plots also reveal that the performance of BLMC is sensitive to the choice of $K$.}
%that the choice of $K$ is important for obtaining reliable inference when using BLMC with a diagonal $\bSigma$ as a factor model. 
We recommend choosing $K$ based on scientific considerations for the problem at hand and exploratory data analyses, or checking the RMSPE value for different $K$ and picking $K$ by an elbow rule \citep{thorndike1953belongs}.

\begin{table} 
	\centering
	\caption{
	{Simulation study summary table 2: posterior mean (2.5\%, 97.5\%) percentiles}} 
%	\begin{ruledtabular}
\scalebox{0.9}{
{
	\begin{tabular}{c|cccccccccc}
	\hline\hline
		K = & 1 & 2 & 3 & 4 \\
		\hline
		%CVG-slope & 0.95 (0.85, 1.0) & 0.94 (0.82, 1.0) & 0.94 (0.82, 1.0) & 0.95 (0.85, 1.0)   \\
		CVGL& 0.28(0.14, 0.93) & 0.39(0.17, 0.94) & 0.49(0.2, 0.95) & 0.58(0.25, 0.96) \\
		CVG & 0.95(0.92, 0.98) & 0.95(0.92, 0.98) & 0.95(0.92, 0.98) & 0.95(0.92, 0.98) \\
		RMSPE& 2.07(1.94, 2.18) & 1.96(1.86, 2.05) & 1.87(1.78, 1.97) & 1.78(1.7, 1.85) \\
		%ESS & 3038.0(194.0, 3813.0) & 2686.0(128.0, 3749.0) & 2384.0(107.0, 3676.0) & 2119.0(96.0, 3611.0) \\
		MCSE & 0.004(0.002, 0.008) & 0.005(0.002, 0.01)  & 0.005(0.003, 0.01) & 0.004(0.003, 0.01) \\
		%$\hat{R}$ & 1.0(1.0, 1.003) & 1.0 (1.0, 1.006) & 1.0 (1.0, 1.015) & 1.0 (1.0, 1.012)\\
	\hline\hline
	    K = & 5 & 6 & 7 & 8 \\
		\hline
		%CVG-slope &  0.94 (0.85, 1.0) & 0.94 (0.85, 1.0) & 0.94 (0.85, 1.0) &  0.95 (0.85, 1.0)\\
		CVGL & 0.67(0.29, 0.96) & 0.74(0.33, 0.96) & 0.81(0.38, 0.96) & 0.86(0.44, 0.96)\\
		CVG & 0.95(0.92, 0.98)& 0.95(0.92, 0.98) & 0.95(0.92, 0.98) & 0.95(0.91, 0.98) \\
		RMSPE &1.7(1.62, 1.77) & 1.63(1.55, 1.69) & 1.56(1.5, 1.63) & 1.51(1.45, 1.57) \\
		%ESS & 1869.0(88.0, 3552.0) & 1622.0(67.0, 3474.0) & 1438.0(54.0, 3389.0) & 1252.0(43.0, 3208.0)\\
		MCSE & 0.005(0.003, 0.01) & 0.005(0.003, 0.013) & 0.005(0.003, 0.011)& 0.005(0.003, 0.011) \\
		%$\hat{R}$ & 1.0 (1.0, 1.021) & 1.0(1.0, 1.027) & 1.0 (1.0, 1.033) & 1.0 (1.0, 1.039) \\
	\hline\hline
	    K = & 9 & 10 & &\\
		\hline
		%CVG-slope & 0.95 (0.85, 1.0) & 0.95 (0.85, 1.0)&  &  \\
		CVGL& 0.91(0.59, 0.96) & 0.95(0.92, 0.96)  & &   \\
		CVG & 0.95(0.91, 0.98)& 0.95(0.91, 0.98)  &  & \\
		RMSPE& 1.46(1.41, 1.51) & 1.43(1.38, 1.48)  & &   \\
		%ESS & 1119.0(42.0, 2948.0) & 1012.0(44.0, 2379.0)  & &   \\
		MCSE & 0.005(0.003, 0.01) & 0.005(0.003, 0.01) & &   \\
		%$\hat{R}$ & 1.01 (1.0, 1.048)  & 1.01 (1.0, 1.052) & &  \\
	\hline\hline
	\end{tabular}
}}	
%	\end{ruledtabular}
	\label{table:sim3_hoff}
\end{table}

\section{Remote-sensed Vegetation Data Analysis}\label{sec: real_data_analy}
We apply our proposed models to analyze Normalized Difference Vegetation Indices (NDVI) and Enhanced Vegetation Indices (EVI) measuring vegetation activity on the land surface, which can help us understand the global distribution of vegetation types as well as their biophysical and structural properties and spatial variations. Apart from vegetation indices, we consider Gross Primary Productivity data, Global Terrestrial Evapotranspiration (ET) Product, and landcover data \citep[see][for further details]{ramon2010modis, mu2013modis, sulla2018user}.
%The Vege Indices data records the standard Normalized Difference Vegetation Index (NDVI) and Enhanced Vegetation Index (EVI). These two indices are robust, empirical measures of vegetation activity at the land surface, that are studied for an understanding of the global distribution of vegetation types as well as their biophysical and structural properties and spatial/temporal variations \citep{ramon2010modis}.
%Provided along with the two vegetation indexes are red reflectance, near-infrared (NIR) reflectance, blue reflectance mid-infrared (MIR) reflectance, view zenith angle, sun zenith angle and relative azimuth angle. The GPP data includes GPP and Net Photosynthesis (PsnNet). %($kg C/m^²$). 
%Similar to Vege Indice data, the GPP data is highly correlated to vegetation type. 
%The ET data contains the average daily global evapotranspiration (ET), latent heat flux (LE), potential ET (PET) and potential LE (PLE) %($10^{-4}$ J/$m^2$/day)
%. We refer readers to \citet{mu2013modis} for more details. %The landcover dataset provides categorical variables for different land types \citep{sulla2018user}. 
The geographic coordinates of our variables were %sampled in square tile units that are approximately 1200-by-1200 km (at the equator), and 
mapped on a Sinusoidal (SIN) projection grid. We focus on zone \textit{h08v05}, which covers 11,119,505 to 10,007,555 meters south of the prime meridian and 3,335,852 to 4,447,802 meters north of the equator. The land is situated in the western United States. Our explanatory variables included an intercept and a binary indicator for no vegetation or urban area through the 2016 land cover data. All other variables were measured through the MODIS satellite over a 16-days period from 2016.04.06 to 2016.04.21. Some variables were rescaled and transformed in exploratory data analysis for the sake of better model fitting.
%The Vege Indices data were measured over a 16-days period from 2016.04.06 to 2016.04.21. We averaged GPP data and ET data from two 8-days periods, 2016.4.06 - 2016.4.13 and 2016.4.14 - 2016.4.21, to obtain an estimation of GPP over period 2016.4.06 - 2016.4.21. 
%The landcover data is measured every year and we used the data for year 2016 in the analysis. 
The data sets were downloaded using the \texttt{R} package \textit{MODIS} and the code for the exploratory data analysis is provided as supplementary material to this paper.  

%There are in total 3,115,934 observed locations. We used the whole dataset to illustrate the conjugate NNGP models. 
Our data comprises 1,020,000 observed locations to illustrate the proposed model. %BLMC and a BLMC model with diagonal $\bSigma$.  
%For the conjugate NNGP models, we chose NDVI and red reflectance as the responses and used the index for no vegetation or urban area land along with the intercept as the explanatory variables. For BSLMC and response NNGP model with misalignment, we follow the same setting of the conjugate NNGP models over the dataset with 1,200,000 observed locations. 
Our spatially dependent outcomes were the transformed NDVI ($\log(\mbox{NDVI} + 1)$ labeled as NDVI) and red reflectance (red refl). A Bayesian multivariate regression model, {defined by \eqref{eq: SLMC_model_element} excluding $\bLambda^\top \fb(\bs)$}, was also fitted for comparisons. 
%Based on the exploratory analysis, we observed two groups of responses that have high within-group correlations but relatively low between-group correlations (see figure~\ref{subfig:corr_resp_factor_mapsa}). %There is a high correlation among NDVI, EVI, GPP, red reflectance, blue reflectance, LE and ET. The remaining responses PLE and PET have high correlation with each other and have a relatively low correlation with the other responses. 
%Hence we picked $K = 2$ for the factor BSLMC model% with diagonal $\bSigma$. 
All NNGP based models used $m = 10$ nearest neighbors.
%For all models, 
We randomly held out 10\% of each response and then held all responses over the region 10,400,000 to 10,300,000 meters south of the prime meridian and 3,800,000 to 3,900,000 meters north of the equator to evaluate the models' predictive performance over a missing region (white square) and randomly missing locations. 
%For conjugate models, there were in total 67,132 locations held for prediction. For BSLMC model and response NNGP model with misalignment, we held 22,146 locations for both responses and 181,389 locations for one of the two responses. For the factor BSLMC model with diagonal $\bSigma$, we held responses at locations 12,057 locations for all responses and 656,366 locations for at least one but not all responses.
Figure~\ref{subfig:real_conj_latent_mapsa} illustrates the map of the transformed NDVI data. %The white square region within the Continent is the region held out for prediction. 
%All models in this paper were run on a same machine.
%All models were run on a Linux environment (Ubuntu 18.04.2 LTS), with 32 GB of RAM and 1 Intel Core i7-7700K CPU @ 4.20GHz processor with 4 cores each and 2 threads per core - totaling 8 possible threads for use in parallel computing. %Convergence diagnostics and other posterior summaries were implemented within the Julia statistical environment.

{We fit both models with 5,000 iterations after 5,000 iterations as burn-in.} %We took the first three quarters as burn-in, and the first quarter were sampled with an adaptive algorithm. % rerun BSLMC model, change 50% burn-in into 75% 
The priors for all parameters except decays followed those in the simulation section. We assigned $\mbox{Gamma}(200, 0.02)$ and $\mbox{Gamma}(200, 0.04)$ for $\phi_1$ and $\phi_2$ for BLMC based on fitted variograms to the raw data. 
%We recursively shrink the domain and the grid of candidate values $\{\phi, \alpha\}$ through repeatedly using cross-validation algorithms for fixing parameters for the response NNGP model with misalignment. 
%The recorded run time for running cross-validation algorithms, therefore, varied a lot across different models. 
%The number of threads used in the cross-validation algorithms for response NNGP models with misalignment were equal to the number of folders. The remaining part of 
All the code were run with single thread. No other processes were simultaneously run so as to provide an accurate measure of computing time. 

%The results for the conjugate models are listed in Table~\ref{table:real_conj}.
Table~\ref{table:real_BSLMC} presents results on the BLMC. %and response NNGP with misalignment. 
%Consistent with the related background, 
The regression coefficients of the index of no vegetation or urban area show relatively low biomass (low NDVI) and high red reflectance over no vegetation or urban area. {Estimates of $\bSigma$ and the finite sample process covariance matrix $\bOmega$, as defined in Section~\ref{subsec: sim_2}, show a negative association between the residuals and latent processes of transformed NDVI and red reflectance}, which satisfies the underlying relationship between two responses. BLMC captured a high negative correlation ($\approx -0.87$) between the latent processes of two responses, indicating that the spatial pattern of the latent processes of NDVI and red-reflectance are almost the reverse of each other. The maps of the latent processes recovered by BLMC, presented in Figure~\ref{fig:real_conj_latent_maps}, also support this relationship. 

\begin{table}[!ht]
\caption{Vegetation data analysis summary table 1: posterior mean (2.5\%, 97.5\%) percentiles}
%\begin{adjustbox}{width=0.85\textwidth}
\begin{minipage}[t]{\textwidth} % <--- new
\centering
\scalebox{0.88}{	
	\begin{tabular}{c|c|cc}
	\hline\hline
	& {Bayesian linear model} & \multicolumn{2}{c}{BLMC} \\
	\hline
		& inference & inference & {MCSE}  \\
	\hline
			$\mbox{intercept}_1$  & {0.2515(0.2512, 0.2517)} &   
			{0.1433(0.1418, 0.1449)} & {1.145e-4} \\
			$\mbox{intercept}_2$ & {0.1395(0.1394, 0.1396)}&  
			{0.1599 (0.159, 0.1608)} & {6.17e-5}  \\
			$\mbox{no vege or urban area}_1$ & 
			{-0.1337( -0.1346, -0.1328)}& 
			{-1.385e-2 (-1.430e-2, -1.342e-2)} & {1.69e-5} \\
			$\mbox{no vege or urban area}_2$ & 
			{6.035e-2 (5.992e-2, 6.075e-2)}&  
			{7.831e-3 (7.584e-3, 8.097e-3)} & {8.24e-6} \\
			$\bSigma_{11}$ & 
			{1.599e-2 (1.594e-2, 1.603e-2)} & 
			{3.514e-4 (3.477e-4, 3.553e-4)} & {1.93e-7} \\
			$\bSigma_{12}$ & {-6.491e-3(-6.512e-3, -6.471e-3)}& 
			{-1.084e-4 (-1.100e-4, -1.067e-4)} & {8.19e-8} \\
			$\bSigma_{22}$ & {3.656e-3(3.646e-3, 3.667e-3)} & 
			{1.074e-4 (1.063e-4, 1.084e-4)}& {4.79e-8} \\
			$\bOmega_{11}$ & --&  
			{1.675e-2(1.674e-2, 1.676e-2)} & {4.17e-7} \\
			$\bOmega_{12}$ &-- & 
			{-6.873e-3(-6.879e-3, -6.867e-3)} & {1.77e-7} \\
			$\bOmega_{22}$ &-- & 
			{3.764e-3 (3.760e-3, 3.768e-3)} & {9.06e-8} \\
			$\phi_1$ &--& 
			{3.995 (3.887,  4.075)} & {7.535e-3 } \\
			$\phi_2$ &--& 
			{12.376 (11.512, 13.320)} & {7.60e-3} \\
			\hline
			%MAE\footnotemark[1] & [0.0582, 0.0281, 0.0432] & [0.0188, 0.01112, 0.01496]  & [0.0191, 0.0109, 0.0150] \\
			RMSPE$^a$
			%\footnotemark[1] 
			& {[0.074, 0.0359, 0.0581]} & {[0.0326, 0.0171, 0.0260]} & \\
			CRPS$^a$
			%\footnotemark[1] 
			& {[-0.04135, -0.01988, -0.03061]} & 
			{[-0.01561, -0.00879, -0.0122]} & \\
			CVG$^a$
			%\footnotemark[1] 
			& {[0.956, 0.958, 0.957]} & {[0.954, 0.947, 0.950]} & \\
			INT$^a$
			%\footnotemark[1]
			& {[0.3468, 0.1711, 0.2589]}& {[0.1965, 0.0995, 0.1480]} &\\
			time(mins)& {10.83}
			%\footnotemark[2]
			& {2317.5} &\\
			\hline\hline
	\end{tabular}
%\end{adjustbox}
\footnotetext[1]{[1st response transformed NDVI, 2nd response red reflectance, all responses]}
%\footnotetext[2]{[time for cross-validation, time for MCMC sampling, time for recovering $\bbeta$ and predictions]}
}
\end{minipage}
\label{table:real_BSLMC}
\end{table}

{We provide %mean absolute error (MAE = $n_{test}^{-1}\sum_{i = 1}^{n_{test}}|y_i - E(y_i \given Y)|$), 
%root mean squared predictive error (
RMSPE, % = $\sqrt{1/n\sum_{i =1} ^{n_{test}} (y_i - E(y_i \given Y))^2}$
%), 
CVG, CRPS, INT, MCSE and run time in Table~\ref{table:real_BSLMC}. %We recorded the performance metrics in the table~\ref{table:real_conj} and table~\ref{table:real_BSLMC}. %, table~\ref{table:real_factor} lists the performance metrics for each response for the factor BSLMC model with diagonal $\bSigma$. %In cases where only a prediction interval was provided, the predictive standard error was taken as $(U - L)/(2 \times \Phi^{-1}(0.975))$ where $U$ and $L$ are the upper and lower ends of the interval, respectively. \textcolor{red}{Add them...}
Apparently BLMC substantially improved predictive accuracy. 
BLMC's %and the response NNGP with misalignment 
RMSPEs were over 50\% less than the Bayesian linear model. CVG is similar between two models, while INT and CRPS also favored BLMC over the Bayesian linear model.} Figure~\ref{fig:real_conj_latent_maps} presents the estimated latent processes from BLMC. Notably, the BLMC smooths out the predictions in the held-out region.  
%The sampling process of the conjugate response and latent models cost 1.8 and 18.88 minutes, respectively, which is quite good regarding the sample size of around 3 million. The run time for both cross-validation algorithm and sampling process for conjugate models is appealing for such a large scale dataset. %Response NNGP model with misalignment took around 4.7 hours for the dataset with around 1 million observations, while 
The model's run time was around 38.6 hours, %Since the Response NNGP model with misalignment is similar in prediction to BSLMC with respect to RMSPE, CRPS, CVG and INT, we recommend the former for just prediction and the latter for a full Bayesian inference including the latent process. 
%The total run time for factor BSLMC model with diagonal $\bSigma$ was around 75 hours (4518.67 minutes). 
which is still impressive given the full model-based analysis it offers for such a massive multivariate spatial data set. %the run time for BLMC model %and factor BSLMC model with diagonal $\bSigma$ is still appealing. 

\begin{figure}[!ht]
     \subfloat[ \label{subfig:real_conj_latent_mapsa}]{%
       \includegraphics[width=0.3\textwidth]{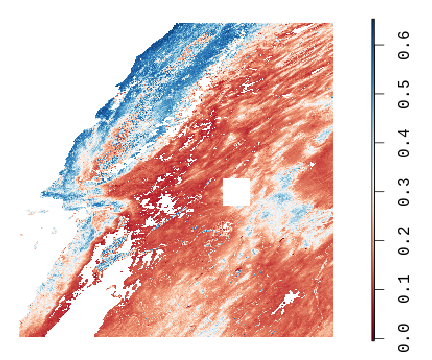}
     }
     \hfill
     \subfloat[ \label{subfig:real_small_mapg}]{%
       \includegraphics[width=0.3\textwidth]{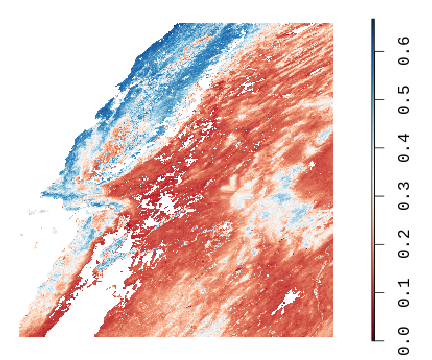}
     }
     \hfill
     \subfloat[ \label{subfig:real_small_mapsh}]{%
       \includegraphics[width=0.3\textwidth]{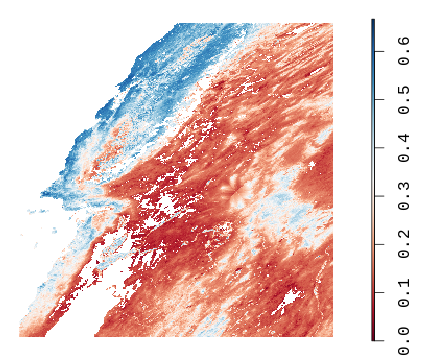}
     }\\
     \subfloat[ \label{subfig:real_conj_latent_mapse}]{%
       \includegraphics[width=0.3\textwidth]{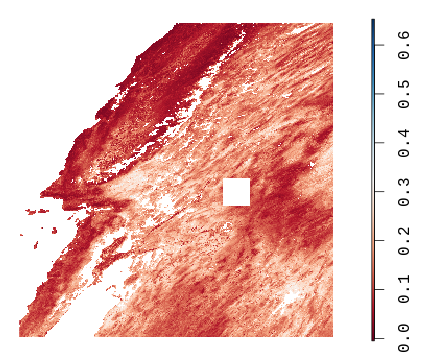}
     }
     \hfill
    \subfloat[ \label{subfig:real_small_mapsi}]{%
       \includegraphics[width=0.3\textwidth]{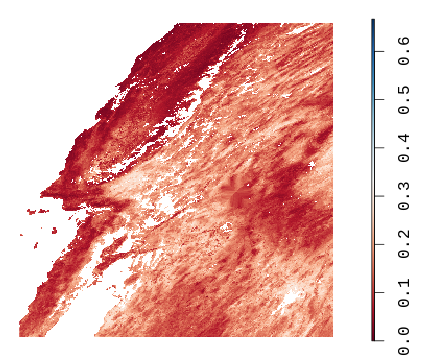}
     }
     \hfill
     \subfloat[ \label{subfig:real_small_mapsg}]{%
       \includegraphics[width=0.3\textwidth]{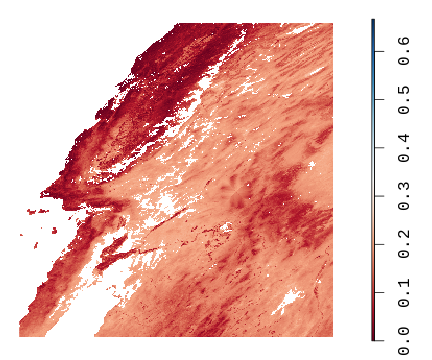}
     }
     \hfill
     \subfloat[ \label{subfig:corr_resp_factor_mapsa}]{%
       \includegraphics[width=0.3\textwidth]{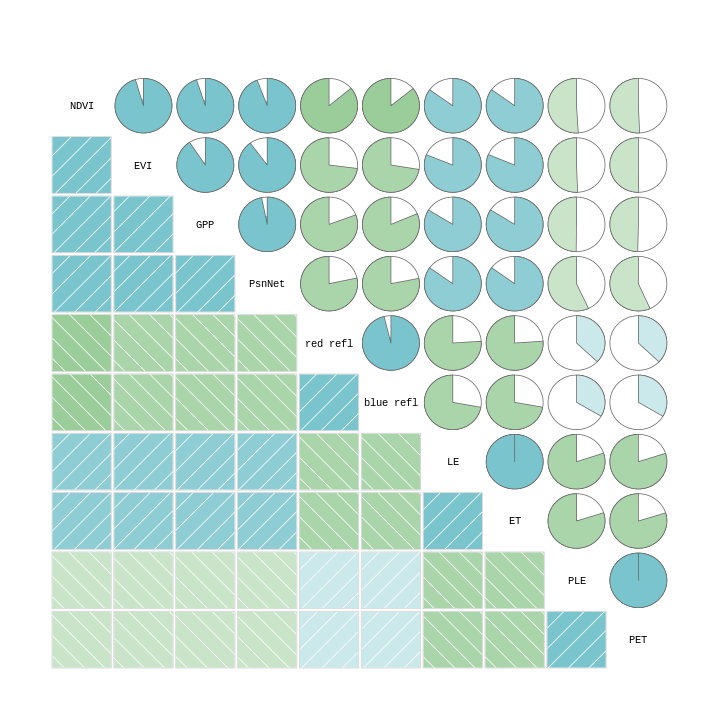}}
       \hfill
     \subfloat[ \label{subfig:corr_resp_factor_mapsb}]{%
       \includegraphics[width=0.3\textwidth]{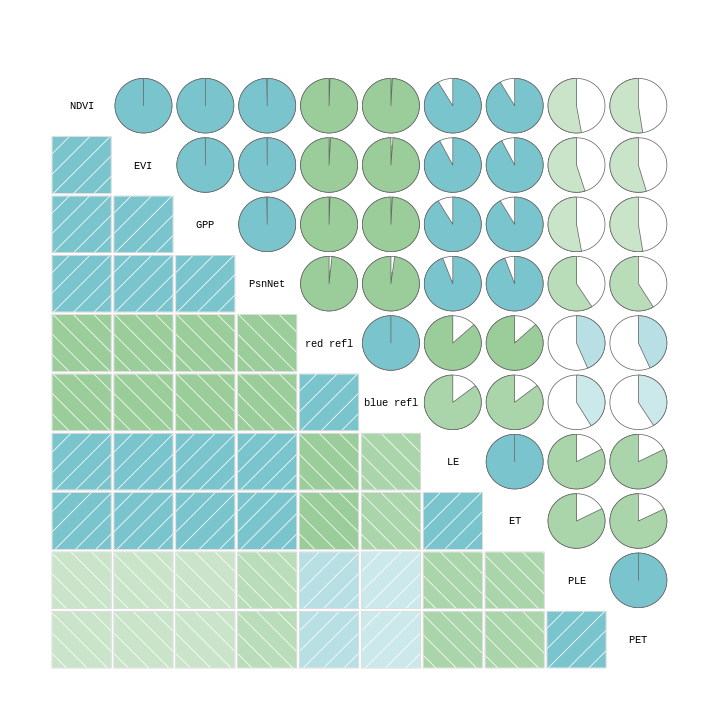}
     }
     \hfill
     \subfloat[ Misalignment \label{subfig20}]{%
       \includegraphics[width=0.35\textwidth]{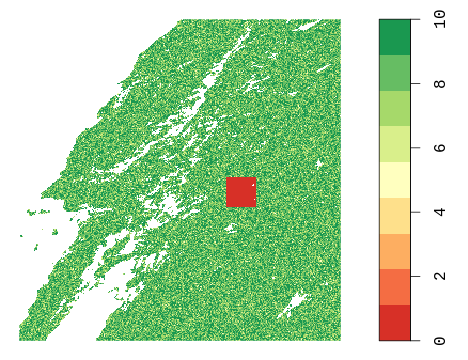}
     }
     \caption{Colored NDVI and red reflectance images of western United States (zone h08v05). Maps of raw data (a) \& (d) and the posterior mean of the intercept-centered latent process recovered from (b) \& (e) BLMC and (c) \& (f) BLMC with diagonal $\bSigma$. Correlation of responses (g) and posterior mean of finite sample correlation among latent processes from the BLMC model with diagonal $\bSigma$ (h).
     Heat-map (i) of counts of observed response. Each observed location is labeled with a color dot whose color represents the count of observed responses on the location. The greener the color, the higher the count. }
%Data from the western United States (tile h08v05) , corresponding to the period from June 25 to July 10, 2000.
\label{fig:real_conj_latent_maps}
\end{figure}

We also fitted a BLMC with diagonal $\bSigma$ to explore the underlying latent processes of ten (transformed) responses: (i) NDVI, (ii) EVI, (iii) Gross Primary Productivity (GPP), (iv) Net Photosynthesis (PsnNet), (v) red reflectance (red refl), (vi) blue reflectance (blue refl), (vii) average daily global evapotranspiration (ET), (viii) latent heat flux (LE), (ix) potential ET (PET) and (x) potential LE (PLE). 
%We held all responses over the region in the previous example and randomly held 10\% of each response to examine the predictive performance. 
There are, in total, $12,057$ locations with no responses and $656,366$ observed locations with misaligned data (at least one but not all responses), which covers $65.12$\% of observed locations. We provide a heat-map (Figure~\ref{subfig20}) to present the status of misalignment over the study domain.

Based on the exploratory analysis, we observed two groups of responses that have high within-group correlations but relatively low between-group correlations (see Figure~\ref{subfig:corr_resp_factor_mapsa}). %There is a high correlation among NDVI, EVI, GPP, red reflectance, blue reflectance, LE and ET. The remaining responses PLE and PET have high correlation with each other and have a relatively low correlation with the other responses. 
Hence we picked $K = 2$. % for the BLMC with diagonal $\bSigma$. 
%\textcolor{blue}{
Estimates from the BLMC model are presented in Table~\ref{table:real_factor}. No vegetation or urban area exhibits lower vegetation indexes (lower NDVI and EVI) and lower production of chemical energy in organic compounds by living organisms (lower GPP and PsnNet). We observe a trend of higher blue reflectance, red reflectance, evapotranspiration (higher ET LE) and lower potential evapotranspiration (lower PET PLE) in urban area and area with no vegetation. We provide maps of posterior predictions for all $10$ variables in Appendix~\ref{SM: predict_map}. 
%}
The latent processes corresponding to transformed NDVI and red reflectance fitted in two analyses 
%through BLMC and the BLMC with diagonal $\bSigma$ 
in Figure~\ref{fig:real_conj_latent_maps} share a similar pattern. 
Finally, the heat map of the posterior mean of the finite sample correlation among the latent processes (elements of $\bOmega$ as defined in Section~\ref{subsec: sim_2}) based on BLMC with diagonal $\bSigma$, presented in Figure~\ref{subfig:corr_resp_factor_mapsb}, reveals a high underlying correlation among NDVI, EVI, GPP, PsnNet, red and blue reflectance, and that LE and ET are slightly more correlated with NDVI and EVI than PLE and PET. 
The total run time for BLMC with diagonal $\bSigma$ was around {60.7 hours (3642.25 minutes)}. %Regarding the scale of the multivariate spatial dataset, the run time for factor BSLMC with diagonal $\bSigma$ is still appealing. 

\begin{table}[!ht]
	\centering
	\caption{Vegetation data analysis summary table 2: posterior mean (2.5\%, 97.5\%)} 
%	\begin{ruledtabular}
\scalebox{0.88}{
{
	\begin{tabular}{c|cccc}
	\hline\hline
		response & slope & MCSE  & nugget ($\bSigma_{ii}$) & MCSE \\
		\hline
		NDVI & 
		-0.0120 (-0.0124, -0.0116) & 1.37e-5  &
		7.46e-4 ( 7.42e-4,  7.49e-4)& 7.60e-8 \\
		EVI & -4.38e-3(-4.68e-3, -4.08e-3) & 6.86e-6  & 8.68e-4(8.65e-4, 8.7e-4) & 3.07e-8 \\
		GPP & 
		-0.197(-0.199, -0.194) & 8.31e-5 &
		0.0244(0.0243,  0.0245) & 2.34e-6 \\ 
		PsnNet & 
		-4.48e-3(-5.39e-3, -3.50e-3) & 3.42e-5 & 
		5.34e-3(5.32e-3, 5.36e-3) & 3.50e-7 \\ 
		red refl & 
		4.49e-3 (4.20e-3, 4.77e-3) & 5.11e-6 & 
		9.84e-4( 9.81e-4, 9.87e-4) & 3.13e-8 \\
		blue refl & 
		0.0123 (0.0121, 0.0124) & 2.74e-6  &
		2.60e-4(2.59e-4,  2.61e-4) & 8.81e-9 \\
		LE & 
		0.0908(0.0884, 0.0932) & 1.36e-4  & 
		0.0531 (0.0529, 0.0533) &  2.37e-6 \\
		ET & 
		0.0919 (0.0895, 0.0944) & 1.49e-4  & 
		0.0531(0.053, 0.0533) & 2.27e-6 \\
		PLE & 
		-3.64e-3 ( -3.98e-3, -3.36e-3) & 5.50e-5 & 
		2.095e-5 (2.086e-5, 2.104e-5) & 1.63e-9 \\
		PET & 
		-4.88e-3(-5.99e-3, -3.96e-3) & 1.81e-4 &
		6.50e-5 ( 6.44e-5, 6.57e-5) & 2.20e-8 \\
		\hline
		%MAE& 0.0252& 0.021& 0.1562& 0.0615& 0.0263& 0.0132& 0.2143& 0.2137& 0.009& 0.0291 &0.077 \\
        %INT&& & & & & & & & &  & \\
	\end{tabular}
}
}
%	\end{ruledtabular}
	\label{table:real_factor}
\end{table}

\section{Summary and Discussion}\label{sec: summary}
We have proposed scalable models for analyzing massive and possibly misaligned multivariate spatial data sets. %We obtain the posterior distribution of the high-dimensional multivariate latent spatial process that would be precluded by usual multivariate GPs. 
Our framework offers flexible covariance structures and scalability by modeling the loading matrix of spatial factors using Matrix-Normal distributions and the factors themselves as NNGPs. This process-based formulation allows us to resolve spatial misalignment by fully model-based imputation. Through a set of simulation examples and an analysis of a massive misaligned data set comprising remote-sensed variables, we demonstrated the inferential and computational benefits accrued from our proposed framework. 

%We devoted a significant part of the paper to formulate and illustrate our models in conjunction with NNGP. Meanwhile, we took an elaborate design in our simulation studies to test the performance of all models using data sets with different behaviors. We also showed our models are capable of conducting various analyses for massive multivariate spatial data through implementations on a real dataset with observed locations in millions.

This work can be expanded further in at least two important directions. The first is to extend the current methods to spatiotemporal data sets, where multiple variables are indexed by spatial coordinates, as considered here, as well as by temporal indices. Associations are likely to be exhibited across space and time as well as among the variables within a location and time-point. In addition, these variables are likely to be misaligned across time and space. Regarding the scalability of the spatiotemporal process, we can build a dynamic nearest-neighbor Gaussian process (DNNGP) \citep{datta2016nonseparable} to model spatiotemporal factors and one can also envisage temporal dependence on the loading matrix.

A second direction will consider spatially-varying coefficient models. We model the regression coefficients $\bbeta$ using a spatial (or spatiotemporal) random field to capture spatial (or spatiotemporal) patterns in how some of the predictors impact the outcome. We can assign the prior of the regression coefficients $\bbeta$ using a multivariate Gaussian random field with a proportional cross-covariance function. Then the prior of $\bbeta$ over observed locations follows a Matrix-Normal distribution, which is the prior we designed for $\bbeta$ in all of the proposed models in this article. While the modification seems to be easy, the actual implementation requires a more detailed exploration, and we leave these topics for further explorations. 

%The scalability of proposed algorithms is guaranteed when the utilized univariate scalable modeling method can yield a sparse precision matrix. Hence our approach can adapt to other methods such as multiresolution approximation (MRA), NNGP, spatial partitioning, etc. \cite[see, e.g.][]{katzfuss2017general}. For brevity, we elaborated on the algorithms for NNGP based models in the supplementary material, but provided a limited discussion on alternative models. Further exploration with various univariate scalable modeling methods still require careful discussions. Also, scalable modeling methods that introduce sparsity in covariance matrices such as covariance tapering are beyond the scope of our discussion.

From a computational perspective, we clearly need to further explore high-performance computing and high-dimensional spatial models amenable to such platforms. The programs provided in this work are for illustration and have limited usage in Graphical Processing Units (GPU) computing and parallelized CPU computing. %Since the algorithms for conjugate models are parallelizable, GPU parallel computing can dramatically reduce the run time. 
A parallel CPU computing algorithm for the BLMC model can simultaneously sample multiple MCMC chains, improving the performance of the actual implementations. Implementations with modeling methods such as MRA \citep{katzfussmultires} also requires dedicated programming with GPU. %\textcolor{orange}{From Lu: The following two sentences might need modifications in the last stage.} 
Other scalable modeling methods that build graphical Gaussian models on space, time and the number of variables can lead to sparse models for high-dimensional multivariate data and scale not only up to millions of locations and time points, but also to hundreds or even thousands of spatially or spatiotemporally oriented variables. The idea here will be to extend current developments in Vecchia-type models to graphs building dependence among a large number of variables so that the precision matrices across space, time and variables is sparse. Research on scalable statistical models and high-performance computing algorithms for such models will be of substantial interest to statisticians and environmental scientists. %Efficient programming and developing related packages will be of interest in future research.

%  This section is optional.  Here is where you will want to cite
%  grants, people who helped with the paper, etc.  But keep it short!
\section*{Acknowledgements}
The work of the authors have been supported in part by National Science Foundation (NSF) under grants NSF/DMS 1916349 and NSF/IIS 1562303, and by the National Institute of Environmental Health Sciences (NIEHS) under grants R01ES030210 and 5R01ES027027.

\section*{Supporting information}\label{sec: SM}
The MODIS vegetation indices data analyzed in Section~\ref{sec: real_data_analy}, and the Julia code implementing our models are available are available at \url{https://github.com/LuZhangstat/Multi_NNGP}.%with this paper at the Biometrics website on Wiley Online Library.

\appendix

\bibliographystyle{ba}  
\bibliography{lubib} 

\renewcommand{\theequation}{S.\arabic{equation}}
\renewcommand{\thesection}{S.\arabic{section}}
\renewcommand{\thetheorem}{S.\arabic{theorem}}
\renewcommand{\theremark}{S.\arabic{remark}}

\section{Algorithm of NNGP based BLMC model}\label{SM: BSLMC_NNGP_alg}
We discuss the posterior predictions before going to the detailed algorithm. We use $N_m(\bu_i)$ to denote the $m$ neighbors of $\bu_i \in \calU$ among $\calS$. The posterior prediction for $\fb_k(\calU)$ given in \eqref{eq: FU_LMC_post} follows
\begin{equation}\label{eq: LMC_FU_distr}
\begin{aligned}
    \fb_k(\calU) \given \fb_k, \psi_k &\sim  \mbox{N}(\tilde{\bA}\fb_k, \Tilde{\bD})\;,
    %\\\Tilde{\bA}_i[Pa[\bu_i]] &= \{\rho_{\psi_k}(\bu_i, \calS_{\mbox{Pa}[\bu_i]})\rho_{\psi_k}(\calS_{\mbox{Pa}[\bu_i]}, \calS_{\mbox{Pa}[\bu_i]})^{-1}\}^\top\; ,\\
    %\Tilde{\bD}_i &= 1 -\rho_{\psi_k}(\bu_i, \calS_{\mbox{Pa}[\bu_i]})\rho_{\psi_k}(\calS_{\mbox{Pa}[\bu_i]}, \calS_{\mbox{Pa}[\bu_i]})^{-1}\rho_{\psi_k}(\calS_{\mbox{Pa}[\bu_i]}, \bu_i)\; ,\\
    %\mbox{Pa}[\bu_i] &:= 
	%\mbox{the index of $m$ closest points to $\bu_i$ among $\calS$} \; ,\\
	%\Tilde{\bA} &= [\Tilde{\bA}_1: \cdots: \Tilde{\bA}_n]^\top\;,\; \Tilde{\bD} = \mbox{diag}(\{\Tilde{\bD}_i\}_{i = 1}^n)\; .
\end{aligned}
\end{equation}
where the $(i,j)$-th entry of $\Tilde{\bA}$ is $0$ when $\bs_j \notin N_m(\bu_i)$, and, similar to $\bA_{\rho_k}$, the $m$ nonzero entries in the $i$-th row of $\Tilde{\bA}$ corresponds to the elements of the $1 \times m$ vector \\ $\tilde{\ba}_i^\top = \brho_{\psi_k}(\bu_i, N_m(\bu_i))\brho_{\psi_k}(N_m(\bu_i), N_m(\bu_i))^{-1}$. The $(i,i)$-th diagonal element of $\tilde{\bD}$ equals $\rho_{\psi_k}(\bu_i, \bu_i) - \tilde{\ba}_i^\top\brho_{\psi_k}(N_m(\bu_i), \bu_i)$.
And the posterior sample of $\bY_\calU$ after giving posterior sample of $\bbeta, \bLambda, \bSigma$ and $\bF_\calU$ can be sampled through 
\begin{equation}
    \mbox{MN}(\bX_\calU\bbeta + \bF_\calU\bLambda, \bI_{n'}, \bSigma)\;.
\end{equation}
The following gives the detailed algorithm.\\
\noindent{ \rule{\textwidth}{1pt}
{\fontsize{8}{8}\selectfont
	\textbf{Algorithm~1}: 
	Obtaining posterior inference of $\{\bgamma, \bSigma, \bomega\}$ and predictions on a new set $\calU$ for NNGP based BLMC model \\
	%[-2pt]
	\rule{\textwidth}{1pt}%\\[-12pt]
	\begin{enumerate}%[Step 1:]
	    \item Precalculation and preallocation for the MCMC algorithm
	    \begin{enumerate}
	        \item Find location sets $\calS$, $\calM$ and the index of the observed and missing response $\{os_i\}_{i = 1}^n$ and $\{ms_i\}_{i = 1}^n$.
	        \item Build the nearest neighbor for $\calS$
	        \item Calculate Cholesky decompositions $\bV_{\bLambda} = \bL_{\bLambda}\bL_{\bLambda}^\top$ and $\bV_{\bbeta} = \bL_{\bbeta}\bL_{\bbeta}^\top$
	        \item Preallocate MCMC samples and initalize MCMC chain with $\bbeta^{(0)}$, $\bLambda^{(0)}$, $\bSigma^{(0)}$ and $\{\psi_k^{(0)}\}_{k = 1}^K$
	    \end{enumerate}

	    \item Block update MCMC alogrithm. For $l = 1 : L$
	    \begin{enumerate}
	        \item Update $\bF^{(l)}$ and impute missing response $\{\by(\bs_i)_{mi}^{(l)}\}_{\bs_i\in \calM}$
	        \begin{itemize}
	            \item Construct $\Tilde{\bX}$ and $\Tilde{\bY}$ in \eqref{eq: argument_liner_LMC}
	            \begin{itemize}
	                \item Build the matrix $\bD_{\bSigma_o}^{\frac{1}{2}} = \mbox{diag}(\{\bSigma_{os_i}^{-\frac{1}{2}}\}_{i = 1}^n)$ in  \eqref{eq: argument_liner_LMC}
	                \hfill{$\mathcal{O}(n)$}
	                \item Construct $\{\bA_{\rho_k}\}_{k = 1}^K$ and $\{\bD_{\rho_k}\}_{k = 1}^K$ %in Section~\ref{sec: implement_NNGP_BSLMC_models} 
	                as described, for example, in \citet{finley2019efficient} \hfill{$\mathcal{O}(Knm^3)$}
	                \item Construct $\Tilde{\bX}$ and $\Tilde{\bY}$ in \eqref{eq: argument_liner_LMC} with $\bV_k = \bD_{\rho_k}^{-\frac{1}{2}}(\bI - \bA_{\rho_k})$ \hfill{$\mathcal{O}(nK(m + 1 + q) + npq)$}
	            \end{itemize}
	            \item Use LSMR \citep{fong2011lsmr} to generate sample of $\bF^{(l)}$ 
	            \begin{itemize}
	                \item Sample $\bu \sim \mbox{N}(\mathbf{0}, \bI_{Kn})$ \hfill{$\mathcal{O}(nK)$}
	                \item Solve $\mbox{vec}(\bF)^{(l)}$ from $\tilde{\bX}\mbox{vec}(\bF)^{(l)} = \tilde{\bY} + \bu$ by LSMR
	            \end{itemize}
	            \item Impute missing response $\{\by(\bs_i)_{ms_i}^{(l)}\}_{\bs_i\in \calM}$ over $\calM$ through \eqref{eq: SLMC_YM_cond_post}
	            \begin{itemize}
	                \item Calculate $\bmu_{\bs} = \bbeta^{(l)\T} \bx(\bs) + \bLambda^{(l-1)} \fb(\bs)$ for $\bs \in \calM$
	                \item Sample $\by(\bs)_{ms}^{(l)}$ by \eqref{eq: SLMC_YM_cond_post} for $\bs \in \calM$
	            \end{itemize}
	        \end{itemize}
	        \item Use MNIW to update $\{\bbeta^{(l)}, \bLambda^{(l)}, \bSigma^{(l)}\}$ 
	        \begin{itemize}
	            \item Construct $\bX^{\ast}$ and $\bY^{\ast}$ in \eqref{eq: augment_linear_LMC}
	            \item Generate $\bSigma^{(l)}$
	            \begin{itemize}
	            \item (When $\bSigma$ is a positive symmetric matrix)
	            \begin{itemize}
	                \item Calculate $\bmu^\ast$, $\bV^{\ast-1}$, $\bPsi^\ast$ and $\nu^\ast$ by \eqref{eq: SLMC_MNIW} \hfill{$\mathcal{O}(n(p+K)(p+K+q))$}
	                \item Sample $\bSigma^{(l)}$ from $\mbox{IW}(\bPsi, \nu^\ast)$ 
	            \end{itemize}
	            \item (When $\bSigma$ is diagonal) 
	            \begin{itemize}
	                \item Calculate $\bmu^\ast$ by \eqref{eq: SLMC_MNIW} \hfill{$\mathcal{O}(n(p+K)(p+K+q))$}
	                \item Sample elements of $\bSigma^{(l)}$ from Inverse-Gamma with parameters provided in \eqref{eq: Diag_Sigma_LMC_post}
	            \end{itemize}
	            \end{itemize}
	            \item Sample $\bgamma^{(l)} = [\bbeta^{(l)\top}, \bLambda^{(l)^\top}]^\top$ from $\mbox{MN}(\bmu^\ast, \bV^\ast, \bSigma^{(l)})$
	            \begin{enumerate}
	                \item Sample $\bu \sim \mbox{MN}(\mathbf{0}, \bI_{p+K}, \bI_{q})$
	                \item Calculate Cholesky decomposition  $\bV^{\ast-1} = \bL_{\bV}\bL_{\bV}^\top$ and $\bSigma^{(l)} = \bL_{\bSigma^{(l)}}\bL_{\bSigma^{(l)}}^\top$
	                \item Generate $\bgamma^{(l)} = \bmu^\ast +\bL_{\bV}^{-\top}\bu\bL_{\bSigma^{(l)}}^\top$
	            \end{enumerate}
	        \end{itemize}
	        \item Use Metropolis random walk to update $\{\Psi_k^{(l)}\}_{k = 1}^K$ 
	        \begin{enumerate}
	            \item Propose new $\{\Psi_k^\ast\}_{k = 1}^K$ based on $\{\Psi_k^{(l - 1)}\}_{k = 1}^K$
	            \item Calculate the likelihood of the new proposed  $\{\Psi_k^\ast\}_{k = 1}^K$ and $\{\Psi_k^{(l - 1)}\}_{k = 1}^K$ given $\bF^{(l)}$ using \eqref{eq: SLMC_psi_cond_post}  \hfill{$\mathcal{O}(Knm^3)$}
	            \item Accept the new $\{\Psi_k^\ast\}_{k = 1}^K$ as $\{\Psi_k^{(l)}\}_{k = 1}^K$ with the probability of the ratio of the likelihood of $\{\Psi_k^\ast\}_{k = 1}^K$ and $\{\Psi_k^{(l - 1)}\}_{k = 1}^K$. Let $\{\Psi_k^{(l)}\}_{k = 1}^K$ = $\{\Psi_k^{(l - 1)}\}_{k = 1}^K$ when the new proposal is rejected. 
	        \end{enumerate}
	    \end{enumerate}

		\item Generate posterior samples of $\{\bF_\calU^{(l)}, \bY_\calU^{(l)}\}$ on a new set $\calU$
		\begin{enumerate}
		\item Construct $\Tilde{\bA}$ and $\Tilde{\bD}$ in \eqref{eq: LMC_FU_distr} \hfill{$\mathcal{O}(n'm^3K)$}
		%\item Sample $\bu \sim \mbox{MN}(\mathbf{0}, \bI_{n'}, \bI_q)$ 
		\item Generate $\fb_k(\calU)^{(l)}\sim \mbox{N}(\Tilde{\bA}\fb_k, \Tilde{\bD})$ for $k = 1,\ldots, K$ \hfill{$\mathcal{O}(n'Km)$}
		\item Sample $\bY_\calU^{(l)} \given \bomega_\calU^{(l)}, \bgamma^{(l)}, \bSigma^{(l)}, \bF_\calU^{(l)} \sim \mbox{MN}(\bX_\calU\bbeta + \bF_\calU\bLambda, \bI_{n'}, \bSigma^{(l)})$ \\
		\begin{itemize}
		    \item Sample $\bu \sim \mbox{MN}(\mathbf{0}, \bI_{n'}, \bI_{q})$ \hfill{$\mathcal{O}(n'q)$}
		    \item Generate $\bY_\calU^{(l)} = \bX_\calU\bbeta + \bF_\calU\bLambda + \bu\bL_{\bSigma^\top_{(l)}}$ with $\bF_\calU^{(l)} = [f_1(\calU)^{(l)}: \cdots:f_K(\calU)^{(l)}]$ \hfill{$\mathcal{O}(n'(pq + Kq + q^2)$}
		\end{itemize}
		
		\end{enumerate}
	\end{enumerate}	
	\vspace*{-8pt}
	\rule{\textwidth}{1pt}
} }

\section{Technical details and proofs of results in Section~\ref{sec: Theory}}\label{SM: proofs}
Let us begin with a representation of posterior distributions of the latent model in Section~\ref{sec: Theory}. Let $\bV_{\brho}$ be a non-singular square matrix such that $\brho_{\psi}^{-1}(\calS, \calS) = \bV_{\brho}^\top\bV_{\brho}$. Treat the prior of $\bgamma$ as additional observations and recast $p(\bY, \bgamma \given \bSigma) = p(\bY \given \bgamma, \bSigma) \times p(\bgamma \given \bSigma)$ into an augmented linear model
\begin{equation}\label{eq: augment_linear_latent}
\begin{array}{c}
\underbrace{ \left[ \begin{array}{c} \sqrt{\frac{\alpha}{1 - \alpha}} \bY\\ \bL_r^{-1} \bmu_{\bbeta} \\ \mathbf{0} \end{array} \right]}_{\bY^{*}}
= \underbrace{ \left[ \begin{array}{cc} \sqrt{\frac{\alpha}{1 - \alpha}} \bX & \sqrt{\frac{\alpha}{1 - \alpha}} \bI_n \\ \bL_r^{-1}& \mathbf{0} \\  \mathbf{0}& \bV_{\brho} \end{array} \right] }_{\bX^{*}}
\underbrace{ \left[ \begin{array}{c} \bbeta \\ \bomega \end{array} \right]}_{\bgamma}+ \underbrace{ \left[ \begin{array}{c} \etab_1 \\ \etab_2 \\ \etab_3 \end{array} \right]}_{\etab}%\\
%\hspace{1cm} \bY^{*} \hspace{1cm} = \hspace{1.7cm}\bX^{*} \hspace{2.4cm}\bgamma \hspace{0.5cm}+ \hspace{0.5cm}\etab \hspace{0.8cm}
\end{array} ,
\end{equation}
where $\bL_r$ is the Cholesky decomposition of $\bV_r$, and $\etab \sim \mbox{MN}(\mathbf{0}, \bI_{2n + p}, \bSigma)$. When having a flat prior for $\bbeta$, $\bL_r^{-1}$ degenerates to a zero matrix, showing no information from $\bbeta$'s prior contributes to the linear system. The expression in \eqref{eq: augmented_conj_post_v1_pars} can be simplified as
\begin{equation}\label{eq: augmented_conj_post_v2_pars}
\begin{aligned}
\bV^\ast &= (\bX^{\ast\top}\bX^\ast)^{-1}\; , \;
\bmu^\ast = (\bX^{\ast\top}\bX^\ast)^{-1}\bX^{\ast\top}\bY^\ast \; ,\\
\bPsi^\ast &= \bPsi + (\bY^\ast - \bX^\ast \bmu^\ast)^\top(\bY^\ast - \bX^\ast \bmu^\ast)\; , \;
\nu^\ast = \nu + n \; .
\end{aligned}
\end{equation}

We explore the behavior of the above posterior density as the number of observations becomes large under a true data generating distribution. Assume that the true distribution of the dependent variables is included in the parametric family $f(\mathbf{Y}) = p(\mathbf{Y} \given \boldsymbol{\beta}_0, \boldsymbol{\Sigma}_0)$ for some $\boldsymbol{\Sigma}_0$ and $\boldsymbol{\beta}_0$. For distinguishing the variables based on the number of observations, we make the dependence upon $n$ explicit. Denote $\bX(n)_{n \times p} = [\bx(\bs_1): \cdots: \bx(\bs_n)]^\top$, $\bY(n)_{n \times q} = [\by(\bs_1): \cdots: \by(\bs_n)]^\top$,  $\calS(n) = \{\bs_1, \ldots, \bs_n\}$, $\bkappa(n) = \bC(\calS(n), \calS(n)) + (\alpha^{-1} - 1)\bI_n$. $\bX^\ast(n)$ and $\bY^\ast(n)$ are $\bX^\ast$ and $\bY^\ast$ in \eqref{eq: augment_linear_latent} using $\bX(n)$ and $\bY(n)$ instead of $\bX$ and $\bY$. {In the following results, we denote $\mathbf{J}(n) =\bX(n)^\top \bkappa(n)^{-1}\bX(n)$, $\bA \geq \mathbf{B}$ to mean that $\mathbf{A} - \mathbf{B}$ is a positive semi-definite matrix, and $\mathbf{A}_{ij}$ to be the $(i,j)$-th element of $\mathbf{A}$.}

\begin{lemma}\label{lemma1}
The matrix $\bSigma$ in the conjugate multivariate models is posterior consistent
if and only if $\boldsymbol{\Psi}^\ast(n)_{ij}/ n \to \{\bSigma_{0}\}_{ij} \; a.s.$ for $1 \leq i, j \leq q$ with  $\boldsymbol{\Psi}^\ast(n)$ defined by 
\eqref{eq: collapsed_spatial_pars1} \& \eqref{eq: augmented_conj_post_v1_pars}
%, which can be reformatted as the  $\boldsymbol{\Psi}^\ast(n)$ defined in \eqref{eq: augmented_conj_post_v2_pars}
\end{lemma}

\begin{proof}
The conjugate multivariate models yield $\bSigma \given \bY(n) \sim \mbox{IW}(\boldsymbol{\Psi}^\ast(n), \nu^\ast(n))$ with
$$
\mathbf{M}_{ij} = \mathbb{E}(\bSigma_{ij} \given \bY(n)) =\frac{\bPsi^\ast(n)_{ij}}{c - 1}\; , \; 
\mbox{Var}(\bSigma_{ij} \given \bY(n)) = \frac{(c + 1)\bPsi^\ast(n)^2_{ij} + (c - 1)\bPsi^\ast(n)_{ii}\bPsi^\ast(n)_{jj}}{c(c - 1)^2(c-3)}\;,
$$
where $c = \nu^\ast(n) - p$, $\boldsymbol{\Psi}^\ast(n)$ and $\nu^\ast(n)$ are defined in \eqref{eq: collapsed_spatial_pars1} and \eqref{eq: augmented_conj_post_v1_pars} for the response and latent process models, respectively. 

Necessity: If $\bSigma$ is posterior consistent, i.e., for any $\epsilon > 0$
$$
\mbox{lim}_{n \to \infty}\mbox{Pr}(|\bSigma_{ij} - \bSigma_{0ij}| > \epsilon \given \bY(n)) = 0\; \text{for } 1 \leq i, j \leq q\;,
$$ 
then $\mbox{lim}_{n \to \infty}\mathbb{E}(\bSigma_{ij} - \bSigma_{0ij} \given \bY(n)) \leq  \mbox{lim}_{n \to \infty} \mathbb{E}(|\bSigma_{ij} - \bSigma_{0ij}| \given \bY(n)) < \epsilon$ for any $\epsilon > 0$; hence, $\mbox{lim}_{n \to \infty}\mathbb{E}(\bSigma_{ij} \given \bY(n)) = \bSigma_{0ij}$ a.s. Therefore, $\boldsymbol{\Psi}^\ast(n)_{ij}/ n \to \bSigma_{0ij} \; a.s.$ for $1 \leq i, j \leq q$.
%$\mbox{lim}_{n \to \infty}\mbox{Var}(\bSigma_{ij} \given \bY(n)) \leq \mbox{lim}_{n \to \infty} \mathbb{E}[(\bSigma_{ij} - \bSigma_{0ij})^2 \given \bY(n)]\leq \epsilon^2$ for any $\epsilon > 0$. Therefore, we have $\mbox{lim}_{n \to \infty}\mbox{Var}(\bSigma_{ij} \given \bY(n)) = 0$ a.s. and $\mbox{lim}_{n \to \infty}\mathbb{E}(\bSigma_{ij} \given \bY(n)) = \bSigma_{0ij}$ a.s.

%$\boldsymbol{\Sigma} \given \bY(n) \overset{p}{\rightarrow} \boldsymbol{\Sigma}_0$ implies that $\mathrm{lim}_{n \rightarrow \infty}\hat{\boldsymbol{\Sigma}}(n) = \mathrm{lim}_{n \rightarrow \infty}\boldsymbol{\Psi}^\ast(n) / n  \rightarrow \boldsymbol{\Sigma}_0$

Sufficiency: If $\boldsymbol{\Psi}^\ast(n)_{ij}/ n \to \bSigma_{0ij} \; a.s.$ for $1 \leq i, j \leq q$, then, from the posterior distribution of $\bSigma$ directly obtain $\mbox{lim}_{n \to \infty}\mathbb{E}(\bSigma_{ij} \given \bY(n)) = \bSigma_{0ij}$ and the variance of each element converges to $0$ at the rate of $1/n$. Using the triangle and Chebyshev's inequalities, for any $\epsilon > 0$ we obtain $\displaystyle \mbox{Pr}(|\bSigma_{ij} - \bSigma_{0ij}| > \epsilon \given \bY(n)) \leq \mbox{Pr}(|\bSigma_{ij} - \mathbf{M}_{ij}| > \epsilon/2 \given \bY(n)) + \mbox{Pr}(|\mathbf{M}_{ij} - \bSigma_{0ij}| > \epsilon/2 \given \bY(n)) \leq 4\mbox{Var}(\bSigma_{ij} \given \bY(n)) / \epsilon^2 + \mbox{Pr}(|\mathbf{M}_{ij} - \bSigma_{0ij}| > \epsilon/2 \given \bY(n)) \to 0 \; a.s.$.
%\begin{align*}
%\mbox{Pr}(|\bSigma_{ij} - \bSigma_{0ij}| > \epsilon \given \bY(n)) &\leq \mbox{Pr}(|\bSigma_{ij} - \mathbf{M}_{ij}| > \epsilon/2 \given \bY(n)) + \mbox{Pr}(|\mathbf{M}_{ij} - \bSigma_{0ij}| > \epsilon/2 \given \bY(n)) \\
%&\leq 4\mbox{Var}(\bSigma_{ij} \given \bY(n)) / \epsilon^2 + \mbox{Pr}(|\mathbf{M}_{ij} - \bSigma_{0ij}| > \epsilon/2 \given \bY(n)) \to 0 \; a.s.
%\end{align*}
%\textcolor{blue}{Since lemma~\ref{lemma1} holds for conjugate multivariate response model, it also holds for conjugate multivariate latent model with $\boldsymbol{\Psi}^\ast(n)$ defined by \eqref{eq: augmented_conj_post_v1_pars}.}
\end{proof}

\begin{theorem}\label{thmS1}
The matrix $\boldsymbol{\Sigma}$ in the conjugate multivariate models is posterior consistent.
\end{theorem}
\begin{proof}
From \eqref{eq: augment_linear_latent}, it follows that $\mathbf{u}(n) \given \boldsymbol{\Sigma}_0 \sim \mathrm{MN}_{ (2n+p) \times q}(\mathbf{0}, \mathbf{I}_{2n + p}, \boldsymbol{\Sigma}_0)$, %$\bu_j$ denote the $j$-th column of $\bu(n)$ %$[\bu_1, \ldots, \bu_{2n + p}]^\top$ 
where $\mathbf{u}(n) = \mathbf{Y}^\ast(n) - \mathbf{X}^\ast(n)\bgamma$ and we write $\boldsymbol{\Psi}^\ast(n) / n =  \boldsymbol{\Psi}/n + \frac{1}{n}\bu(n)^{\top}\left(\bI_n - \bH^{\ast}(n)\right)\bu(n) 
%\frac{1}{n}\mathbf{u}(n)^\top \mathbf{u}(n) - \frac{1}{n} \mathbf{u}(n)^\top \mathbf{X}^\ast(n) (\mathbf{X}^{\ast}(n)^\top\mathbf{X}^\ast(n))^{-1} \mathbf{X}^{\ast}(n)^\top\mathbf{u}(n)
$, where $\bH^{\ast}(n) = \mathbf{X}^\ast(n) \{\mathbf{X}^{\ast}(n)^\top\mathbf{X}^\ast(n)\}^{-1} \mathbf{X}^{\ast}(n)^\top$ is idempotent with rank $p+n$. Writing $\bH^{\ast}(n) = \mathbf{Q}(n)^\top\tilde{\bI}\mathbf{Q}(n)$, where $\mathbf{Q}(n)$ is an orthogonal matrix and $\displaystyle \tilde{\bI} = \begin{bmatrix}
\mathbf{I}_{p+n} & \mathbf{O}\\ 
\mathbf{O} & \mathbf{O}\end{bmatrix}$ and letting $\mathbf{v}(n) = \mathbf{Q}(n)\mathbf{u}$, we obtain $\mathbf{v}(n) \sim \mathbf{MN}_{(2n+p) \times q}(\mathbf{0}, \mathbf{I}_{2n + p}, \boldsymbol{\Sigma}_0)$. Also, we know $\mathrm{lim}_{n \to \infty} \frac{2}{2n} \{\mathbf{u}(n)^\top\mathbf{u}(n)\}_{ij} = 2\bSigma_{0ij}$ a.s. for $1 \leq i,j, \leq q$ and $\displaystyle \mathrm{lim}_{n \rightarrow \infty} \frac{1}{n} \{\mathbf{u}(n)^\top \bH^{\ast}(n)\mathbf{u}(n)\}_{ij} = \mathrm{lim}_{n \rightarrow \infty}\frac{1}{n} \sum_{l = 1}^p \{\bv_l^\top \bv_l\}_{ij} = \bSigma_{0ij}\; a.s.$ from the Khinchin-Kolmogorov strong law of large numbers,
%\begin{align*}
%   \mathrm{lim}_{n \rightarrow \infty} \frac{1}{n} \{\mathbf{u}(n)^\top \mathbf{X}^\ast(n) (\mathbf{X}^{\ast}(n)^\top\mathbf{X}^\ast(n))^{-1} \mathbf{X}^{\ast}(n)^\top\mathbf{u}(n)\}_{ij} = \mathrm{lim}_{n \rightarrow \infty}\frac{1}{n} \sum_{l = 1}^p \{\bv_l^\top \bv_l\}_{ij} = \bSigma_{0ij}\; a.s.
%\end{align*}
where $\bv_j$ is the $j$-th column of $\bv(n)$. %$\{A\}_{ij}$ is the $(i, j)$-th element of matrix $A$.
Hence, $\mathrm{lim}_{n \rightarrow \infty} \bPsi^\ast(n)_{ij} / n = \bSigma_{0ij}$ a.s. and the result follows from Lemma~\ref{lemma1}.
\end{proof}

\begin{theorem}\label{thmS2}
The regression slopes $\bbeta$ is posterior consistent for both conjugate models if and only if
$\mbox{lim}_{n \to \infty}\lambda_{\min}\{\mathbf{J}(n)\} = \infty$, where $\lambda_{\min}\{\mathbf{J}(n)\}$ is the smallest eigenvalue of $\mathbf{J}(n)$. 
\end{theorem}
\begin{proof}
The augmented linear system \eqref{eq: augment_linear_latent} implies that the marginal posterior mean of $\bbeta$ is an unbiased estimator of $\bbeta_0$ with respect to the true distribution of $\bY(n)$. When $\bbeta$ is posterior consistent, %i.e. $\lim_{n \to \infty} \mbox{Pr}(|\bbeta_{ij} - \bbeta_{0ij}| > \epsilon \given \bY(n)) = 0$ for any $\epsilon > 0$ for $1 \leq i \leq p, 1\leq j \leq q$ then 
$\lim_{n \to \infty}\mbox{Var}(\bbeta_{ij} \given \bY(n)) = 0$ a.s with respect to the true distribution of $\bY(n)$. Moreover, $\lim_{n \to \infty}\mbox{Var}(\bbeta_{ij} \given \bY(n)) = 0$ a.s. is a sufficient condition for the posterior consistency of $\bbeta$ through Chebyshev's inequality. In the conjugate model, $\bbeta \given \bY(n) \sim {{\mathrm{T}}}_{{p,q}}(\nu^\ast(n) - q + 1,\bmu^\ast(n), \bV^\ast (n) ,\bPsi^\ast(n))$ with parameters given in \eqref{eq: collapsed_spatial_pars1}. From Theorem~\ref{thm1}, $\lim_{n \to \infty}\bPsi^\ast(n)_{ij} / n = \bSigma_{0ij}$ a.s. for $1 \leq i,j, \leq q$, hence $\lim_{n \to \infty}\mbox{Var}(\bbeta_{ij} \given \bY(n)) = 0$ a.s. if and only if $\lim_{n \to \infty}\{\mathbf{V}^{\ast}(n)\}_{ii} = 0$ for all $i = 1, \ldots, q$. Following \citet{eicker1963asymptotic} (see his proof of Theorem~1), the sufficient and necessary condition is $\lim_{n \to \infty}\lambda_{\min}\{\mathbf{V}^{\ast-1}(n)\}= \infty$. Since $\lambda_{\min}\{\mathbf{J}(n)\} + \lambda_{\max}\{\mathbf{V_r}^{-1}\} \geq \lambda_{\min}\{\mathbf{V}^{\ast-1}(n)\} = \lambda_{\min}\{\mathbf{J}(n) + \mathbf{V_r}^{-1} \} \geq \lambda_{\min}\{\mathbf{J}(n)\}$, the condition simplifies to $\lim_{n \to \infty}\lambda_{\min}\{\mathbf{J}(n)\} = \infty$.
\end{proof}

The following remarks reveal that the posterior consistency in Theorem~\ref{thm1} satisfies with common conditions. 
\begin{remark}\label{remarkS1}
$\lambda_{\min}\{\mathbf{J}(n)\}$ is non-decreasing, and when  $\bbeta$ is posterior consistent, $\lim_{n \to \infty}\mathbf{J}(n)_{ii} = \infty$ since $\mathbf{J}(n)_{ii} \geq \lambda_{\min}\{\mathbf{J}(n)\}$
\end{remark}
%\textbf{Remark 2.1} 
\begin{proof}
Let $\mathbf{X}(n+1) = [\mathbf{X}(n)^\top, x_{n+1}]^\top$, $\bkappa(n+1) = \begin{bmatrix}
\bkappa(n) & \bkappa_{(n), n+1}\\ \bkappa_{n+1, (n)} & \alpha^{-1} \end{bmatrix}$. Then
\begin{equation}
\begin{aligned}
\mathbf{J}(n+1) &= [\mathbf{X}(n)^\top, x_{n+1}] \begin{bmatrix}
\bkappa(n) & \bkappa_{(n), n+1}\\ \bkappa_{n+1, (n)} & \alpha^{-1} \end{bmatrix}^{-1} \begin{bmatrix}
\mathbf{X}(n)\\ x_{n+1}\end{bmatrix} \\
&= \mathbf{J}(n) + \{\mathbf{X}(n)^\top \bkappa(n)^\top \bkappa_{(n), n+1} - x_{n+1}\}d\{\bkappa_{n+1, (n)}\bkappa(n)\mathbf{X}(n) - x_{n+1}\}\\
&=\mathbf{J}(n) + \mathbf{A}(n)\; ,
\end{aligned}
\end{equation}
where $d = \{\alpha^{-1} - \bkappa_{n+1, (n)}\bkappa(n)^{-1}\bkappa_{(n), n+1}\} > 0$ and $\mathbf{A}(n)$ is positive semi-definite symmetric matrix. Thus, $\lambda_{\min}\{\mathbf{J}(n+1)\} = \lambda_{\min}\{\mathbf{J}(n) + \mathbf{A}(n)\} \geq \lambda_{\min}\{\mathbf{J}(n)\}$.
\end{proof}

\begin{remark}\label{remarkS2}
When $\mathbf{X}(n) \sim \mbox{MN}(\mathbf{0}, \bkappa(n), \bSigma^\ast)$ for some 
$\bSigma^\ast$, $\bbeta$ is posterior consistent.
\end{remark}
%\textbf{Remark 2.2} \\
\begin{proof}
Let $\bkappa(n)^{-\frac{1}{2}}$ be the square root of $\bkappa(n)^{-1}$. Then, $\bkappa(n)^{-\frac{1}{2}}\mathbf{X}(n) \sim \mathrm{MN}(\mathbf{0}, \mathbf{I}_n, \boldsymbol{\Sigma}^\ast)$ and the strong law of large numbers ensures $\lim_{n \to \infty} \{\frac{1}{n} \mathbf{J}(n)\}_{ij} %= \lim_{n \to \infty} \{\frac{1}{n} \mathbf{X}^\top(n)\bkappa(n)^{-\frac{1}{2}}\bkappa(n)^{-\frac{1}{2}}\mathbf{X}(n)\}_{ij}%
= \bSigma^*_{ij} \; a.s. \; \text{for}\; 1 \leq i,j \leq p$. Hence, $\lambda_{\min}\{\mathbf{J}(n)\} \rightarrow \infty$.
\end{proof}
%Thus, when the explanatory variables share the same spatial correlation with the responses, the sufficient and necessary condition for Theorem~\ref{thm2} hold.

\begin{remark}\label{remarkS3}
 If $\mathbf{X}(n) \sim \mbox{MN}(\mathbf{0}, \mathbf{I}_n, \bSigma^\ast)$ for some 
$\bSigma^\ast$, then $\bbeta$ is posterior consistent
\end{remark}
%\textbf{Remark 2.3}\\
\begin{proof}
For any $n$, there exists an orthogonal matrix $\mathbf{Q}(n)$ and a diagonal matrix $\mathbf{D}(n)$ with diagonal entries $d_i$ for $i=1,2,\ldots,n$ such that $\mathbf{C}(\calS(n), \calS(n)) = \mathbf{Q}(n)^\top \mathbf{D}(n) \mathbf{Q}(n)$. This yields
\begin{equation}
\begin{aligned}
    \mathbf{J}(n) &= \mathbf{X}(n)^{\top}\mathbf{Q}(n)^\top\{\mathbf{D}(n) + (\alpha^{-1} - 1) \mathbf{I}_n\}^{-1} \mathbf{Q}(n)\mathbf{X}(n)\\
    &= \mathbf{Z}(n)^{\top}\mbox{diag}\left[\{\frac{1}{d_i+ (\alpha^{-1} - 1)}\}_{i = 1}^n\right]\mathbf{Z}(n) = \sum_{i = 1}^n \frac{1}{d_i + (\alpha^{-1} - 1)}\bz_i \bz_i^\top\;,
\end{aligned}
\end{equation}
where $\mathbf{Z}(n) = [\bz_i: \cdots : \bz_n]^\top \sim \mathrm{MN}(\mathbf{0}, \mathbf{I}_n, \boldsymbol{\Sigma}^\ast)$ and $\sum_{i = 1}^n d_i= n$, $d_i \geq 0, i = 1, \ldots, n$. Letting $\bV_i = \bz_i \bz_i^\top$ for $i = 1, \ldots, n$ and applying the matrix version of Cauchy-Schwarz inequality \citep[see, e.g., equation 4 in][]{marshall1990matrix} we obtain
\begin{align*}
\sum_{i = 1}^n \{d_i + (\alpha^{-1} - 1)\} \sum_{i = 1}^n\frac{1}{d_i + (\alpha^{-1} - 1)} \bV_i \geq 
\left\{\sum_{i = 1}^n \sqrt{d_i + (\alpha^{-1} - 1)} \frac{ \bV_i^{\frac{1}{2}}}{\sqrt{d_i + (\alpha^{-1} - 1)}}\right\}^2\;,
\end{align*}
where $\bV_i^{\frac{1}{2}}\bV_i^{\frac{1}{2}} = \bV_i$, and, hence, $\displaystyle \mathbf{J}(n) \geq \alpha\left\{ \frac{\sum_{i = 1}^n \bV_i^{\frac{1}{2}}}{\sqrt{n}}\right\}^2$.
%\begin{equation}\label{eq: Remake2.3_1}
%    \mathbf{J}(n) \geq \alpha\left\{ \frac{\sum_{i = 1}^n \bV_i^{\frac{1}{2}}}{\sqrt{n}}\right\}^2
%\end{equation}
Letting $\bV_i = \bz_i \bz_i^\top = \lambda_i \bu_i \bu_i^\top$ where $\lambda_i = \bz_i^\top \bz_i = \|\bz_i\|^{2}$ and $\bu_i = \frac{\bz_i}{\|\bz_i\|}$, we have $\bV_i^{\frac{1}{2}} = \sqrt{\lambda_i}\bu_i \bu_i^\top$. Changing $n$ into $np$ and rewriting $\sum_{i = 1}^n \bV_i^{\frac{1}{2}}$ into $\sum_{i = 1}^n\sum_{k = 1}^p \bV_{ik}^{\frac{1}{2}}$, where $\{\sum_{k = 1}^p \bV_{ik}^{\frac{1}{2}}\}$ for each $i$ is a full rank $p \times p$ matrix with probability 1, we obtain
\begin{equation}\label{eq: Remake2.3_2}
    \mathbf{J}(np) \geq  \frac{\alpha}{p}\left( \frac{\sum_{i = 1}^n\sum_{k = 1}^p \sqrt{\lambda_{ik}}\bu_{ik}\bu_{ik}^\top }{\sqrt{n}}\right)^2
\end{equation}
We now argue that the smallest eigenvalue of the matrix on the right side goes to infinity as $n \rightarrow \infty$, which will imply that $\lambda_{\min}\{\mathbf{J}(np)\}\rightarrow \infty$. Since $\sum_{k = 1}^p \mathbf{V}_{ik}^{\frac{1}{2}} = \sum_{k = 1}^p \sqrt{\lambda_{ik}}\bu_{ik}\bu_{ik}^\top \sim \mathbf{W}_p(\bSigma^{\ast\frac{1}{2}}, p)$, where $\mathbf{W}_p$ is Wishart distribution, $\{\bu_{i1}, \ldots, \bu_{ip}\}$ make up the bases of the space $\mathbb{R}^p$ with probability 1. For any $\bu \in \mathcal{R}^p, \|\bu\| = 1$, we have 
$$
    \bu^\top \sum_{k = 1}^p \left( \sqrt{\lambda_{ik}} \bu_{ik} \bu_{ik}^\top \right )\bu \geq \underset{k = 1, \ldots, p}{\min}(\sqrt{\lambda_{ik}})
$$
Hence, $\lambda_{\min}(\sum_{i = 1}^n\sum_{k = 1}^p \sqrt{\lambda_{ik}}\bu_{ik}\bu_{ik}^\top) \geq \sum_{i = 1}^n \underset{k = 1, \ldots, p}{\min}(\sqrt{\lambda_{ik}})$. Since $\lambda_{ik} = \|\bz_{ik}\|^2$ where $\bz_{ik} \sim \mbox{N}(\mathbf{0}, \bSigma^\ast)$, $\underset{k = 1, \ldots, p}{\min}\{\sqrt{\lambda_{ik}}\}$ are independent and identically distributed with a positive mean 
$\mathbb{E}(\underset{k = 1, \ldots, p}{\min}\{\sqrt{\lambda_{ik}}\}) = c^\ast > 0$ and a finite variance $\sigma^{2\ast}$. By law of large numbers, we have
$\lim_{n \to \infty}\sum_{i = 1}^n\underset{k = 1, \ldots, p}{\min}\{\sqrt{\lambda_{ik}}\} / n = c^\ast
$ a.s.. Therefore, 
$$
    \lambda_{\min}\left\{\left( \frac{\sum_{i = 1}^n\sum_{k = 1}^p \sqrt{\lambda_{ik}}\bu_{ik}\bu_{ik}^\top }{\sqrt{n}}\right)^2\right\} 
    \geq \frac{1}{n}\left\{\sum_{i = 1}^n \underset{k = 1, \ldots, p}{\min}(\sqrt{\lambda_{ik}})\right\}^2 \to \infty
$$
By \eqref{eq: Remake2.3_2}, $\lim_{n \to \infty}\lambda_{\min}\{\mathbf{J}(n)\} = \infty$.
\end{proof}

\section{Values of parameters in simulation examples}\label{sm: values_examples}

\iffalse
We assessed convergence of MCMC chains by visually monitoring auto-correlations and checking the accuracy of parameter estimates using effective sample size (ESS) \citep[][Sec. 10.5]{gelman2013} and Monte Carlo standard errors (MCSE). \citep{flegal2008markov}, which is given as $\mbox{MCSE}(\theta) \equiv \{V(\theta)/\mbox{ESS}(\theta)\}^{1/2}$, where $V(\theta)$ is the variance of $\theta$ of the posterior samples, $\mbox{ESS}(\theta)$ is the effective sample size (ESS) \citep[][Sec. 10.5]{gelman2013}. The effective sample size were calculated through Julia package "MCMCDiagnostics" \citep{MCMCDiagnostics}.
\fi

\subsection{Values of parameters to generate simulations in simulation example 1}

$$
\bSigma = \begin{bmatrix}0.4 & 0.15 \\
0.15& 0.3\end{bmatrix} \; 
\bbeta = \begin{bmatrix} 1.0 & -1.0 \\
-5.0 & 2.0 \end{bmatrix} \;
\bLambda = \begin{bmatrix} 1.0 & 1.0 \\
0.0 & 2.0 \end{bmatrix}\;
\phi_1 = 6.0\;, \; \phi_2 = 18.0\\
$$

\subsection{Values of parameters to generate simulations in simulation example 2}

\[
\{\bSigma_{ii}\}_{i = 1}^{10} = (0.5, 1, 0.4, 2, 0.6, 2.5, 3.0, 0.45, 1.5, 0.5)
\]

\begin{align*}
\{\phi_k\}_{k = 1}^{50} = 
&(11.36,\; 13.43,\; 10.22,\; 6.87,\; 5.89, \;10.09, \;9.17,\; 2.75,\; 5.35,\; 3.43, \\
&4.09,\; 7.81,\; 12.52,\; 9.54,\; 5.56,\; 7.7,\; 5.44, \;7.49,\; 9.12,\; 5.2,  \\ 
&10.61,\; 5.63, \; 5.5, \; 11.65, \;4.64, \;13.16, \;9.51, \;11.77, \;8.8, \;13.43, \\
&7.89, \;11.62, \;6.4, \;12.95, \;8.48, \;2.5, \;12.95, \;13.42,\; 9.59, \;6.31, \\
&8.98, \;4.57, \;6.63, \;11.25, \;4.43, \;4.94, \;3.3, \;9.66, \;13.5, \;8.7)  
\end{align*}

$$\bbeta = 
\begin{bmatrix}
    1.0 &-1.0& 1.0& -0.5& 2.0& -1.5& 0.5& 0.3& -2.0& 1.5 \\
     -5.0& 2.0& 3.0& -2.0&-6.0 &4.0& 5.0 &-3.0& 6.0& -4.0 \\
     8.0 &6.9 &-12.0& 0.0& -4.0& 7.7& -8.8& 3.3& 6.6 &-5.5 
     \end{bmatrix}$$

\[
\begin{array}{c}
 \{\bLambda_{ij}\}_{1 \leq i \leq 10, 1 \leq j \leq 10 } =  \\
 \begin{bmatrix}
-0.38 & -0.33 &  0.23 & -0.38 & -0.13 &  0.31 &  0.28 &  0.0 &  0.42 &  0.18 \\
 -0.39 & -0.13 & -0.13 & -0.31 & -0.15 &  0.42 &  0.13 & -0.07 &  0.21 &  0.1 \\ 
  0.01 &  0.48 &  0.32 &  0.13 & -0.46 &  0.27 &  0.09 &  0.12 &  0.25 &  0.38 \\
  0.13 &  0.48 & -0.47 & -0.48 & -0.34 & -0.09 & -0.28 & -0.21 & -0.19 & 0.44 \\
 -0.47 &  0.46 &  0.24 & -0.45 &  0.44 & -0.29 & -0.36 & -0.46 &  0.34 & -0.44 \\
  0.21 &  0.12 & -0.46 &  0.29 &  0.36 &  0.17 &  0.03 &  0.2  & -0.12 & -0.23 \\
 -0.2  &  0.48 &  0.18 & -0.1 &   0.13 & -0.13 & -0.41 & -0.04 & -0.07 & -0.22 \\
 -0.19 &  0.28 &  0.47 & -0.42 &  0.17 & -0.18 & -0.03 & -0.13 & -0.04 &  0.3 \\
 -0.04 &  0.27 & -0.23 & -0.07 & -0.09 & -0.39 & -0.48 & -0.27 &  0.19 &  0.21 \\
 -0.03 & -0.18 &  -0.08 & -0.12 &  0.35 &  0.3  & -0.33 &  0.34 &  0.38 &  0.31 
 \end{bmatrix}
\end{array}
\]

\[
\begin{array}{c}
\{\bLambda_{ij}\}_{1 \leq i \leq 10, 11 \leq j \leq 20 } = \\
 \begin{bmatrix}
 -0.15 & -0.43 &  0.27 &  0.18 &  0.38 & -0.4  & -0.27 & -0.3 &  -0.38 & -0.09 \\
  0.28 & -0.24 & -0.32 & -0.08 & -0.01 & -0.31 &  0.2  &  0.31 &  0.11 & -0.38 \\
 -0.38 & -0.42 & -0.16 & -0.37 & -0.22 &  0.09 & -0.08 & -0.07 & -0.33 &  0.01 \\
 -0.42 & -0.22 &  0.44 &  0.09 &  0.25 & 0.12 &  0.1  & -0.33 & -0.41  & 0.42 \\
  0.42 &  0.48 & -0.06 &  0.07 &  0.43 &  0.12 & -0.15 &  0.29 &  0.1 &  -0.32 \\
 -0.15 & -0.03 & -0.42 &  0.01 &  0.05 &  0.33 & -0.46 &  0.12 &  0.22 & -0.44 \\
  0.28 & -0.08 & -0.41 &  0.13 &  0.03 &  0.22 &  0.08 &  0.32 &  0.02 & -0.41 \\
  0.35 & -0.39 &  0.37 & -0.47 & -0.08 & -0.01 &  0.09 &  0.06 & -0.21  & 0.38 \\
 -0.38 &  0.2 & -0.21 &  0.21 & -0.11 &  0.27 &  0.2 &  0.17 & 0.31 & -0.12 \\
  0.36 & -0.09 & -0.16 & -0.06 &  0.43 & -0.04 & -0.07 &  0.4 &  -0.39 & -0.06 
   \end{bmatrix}
\end{array}
\]

\[
\begin{array}{c}
\{\bLambda_{ij}\}_{1 \leq i \leq 10, 21 \leq j \leq 30 } = \\
 \begin{bmatrix}
  0.08 &  0.15 &  0.11 &  0.37 &  0.25 &  0.28  & 0.13 & -0.18 &  0.35 &  0.17 \\
 -0.32 & -0.31 &  0.24 & -0.29 & -0.38 & -0.1 &  -0.19  & 0.18 & -0.37 & -0.34 \\
  0.45 &  0.19 &  0.34 & -0.36 &  0.43 &  0.44 & -0.13 & -0.26 & -0.46 & -0.08 \\
  0.38 & -0.48 & -0.22 & -0.14 &  0.5  &  0.08 &  0.02 & -0.07 &  0.07 & -0.3 \\
 -0.49 & -0.48 &  0.34  & 0.1 &  -0.01 &  0.2  &  0.33 &  0.37 &  0.1 &  0.21 \\
  0.1  &  0.11 & -0.33 & -0.16 &  0.06 &  0.25 & -0.37 & -0.1  & -0.16 & -0.13 \\
  0.45 &  0.02 & -0.21 &  0.16 &  0.37 & -0.2  & -0.44 & -0.37 &  0.46 &  0.25 \\
  0.34 &  0.31 &  0.06 & -0.25 &  0.37 &  0.12 &  0.27 & -0.35 &  0.09 & -0.28 \\
 -0.2  & -0.12 & -0.41 &  0.23 & -0.23 & -0.07 & -0.34 &  0.37 & -0.43 &  0.18 \\
  0.36 &  0.14 &  0.47 &  0.3  &  0.36 & -0.09 &  0.1  & -0.01 &  0.11 &  0.43 
   \end{bmatrix}
\end{array}
\]

\[
\begin{array}{c}
\{\bLambda_{ij}\}_{1 \leq i \leq 10, 31 \leq j \leq 40 } = \\
 \begin{bmatrix}
  0.42 &  0.17 & -0.24 &  0.05 & -0.0  & -0.41 & -0.03 & -0.0  & -0.22 &  0.2 \\
  0.26 & -0.22 &  0.33 & -0.06 & -0.06 & -0.36 & -0.31 &  0.14 & -0.14 & -0.1 \\
  0.09 &  0.43 &  0.04 & -0.35 &  0.42 &  0.19 &  0.33 & -0.12 &  0.4  & -0.32 \\
 -0.13 &  0.36 &  0.02 &  0.02 &  0.34 &  0.06 & -0.32 & -0.47 &  0.02 &  0.34 \\
  0.27 & -0.35 & -0.12 &  0.5  &  0.33 &  0.33 & -0.27 &  0.39 &  0.45 &  0.27 \\
  0.38 &  0.11 &  0.05 &  0.38 & -0.34 & -0.19 & -0.12 &  0.39 &  0.2  &  0.31 \\
  0.16 &  0.31 &  0.02 & -0.43 &  0.13 &  0.33 & -0.34 & -0.1  &  0.41 & -0.46 \\
  0.32 & -0.2  & -0.18 & -0.05 &  0.2  & -0.17 & -0.06 &  0.49 & -0.06 &  0.3 \\
  0.44 & -0.05 &  0.06 & -0.22 & -0.16 & -0.43 &  0.04 & -0.23 & -0.22 &  0.11 \\
 -0.23 & -0.34 &  0.45 & -0.47 &  0.03 & -0.09 & -0.47 &  0.28 &  0.27 & -0.4 
  \end{bmatrix}
\end{array}
\]

\[
\begin{array}{c}
\{\bLambda_{ij}\}_{1 \leq i \leq 10, 41 \leq j \leq 50 } = \\
 \begin{bmatrix}
  0.28 &  0.27 & -0.45 & -0.15 &  0.05 &  0.31 &  0.16 &  0.49 &  0.12 & -0.43 \\
  0.14 & -0.16 &  0.21 & -0.3  &  0.36 & -0.29 &  0.17 &  0.16 & -0.35 & -0.05 \\
 -0.2  & -0.31 &  0.11 & -0.46 &  0.41 &  0.09 & -0.24 & -0.21 &  0.4  & -0.05 \\
  0.44 &  0.44 &  0.41 & -0.22 &  0.36 & -0.45 & -0.19 &  0.46 &  0.49 & -0.28 \\
 -0.34 & -0.5  & -0.33 & -0.37 &  0.33 & -0.31 & -0.37 &  0.05 & -0.38 & -0.14 \\
  0.33 &  0.46 & -0.35 & -0.42 & -0.01 & -0.48 & -0.33 & -0.23 & -0.07 &  0.09 \\
  0.21 & -0.49 & -0.31 & -0.04 &  0.23 &  0.43 &  0.22 &  0.23 & -0.25 & -0.45 \\
 -0.08 &  0.35 &  0.01 &  0.25 & -0.07 & -0.29 & -0.05 &  0.19 & -0.07 & -0.14 \\
 -0.4  & -0.38 &  0.07 &  0.23 &  0.43 & -0.05 &  0.08 &  0.03 &  0.09 &  0.02 \\
 -0.13 & -0.08 & -0.18 & -0.02 & -0.38 & -0.07 &  0.41 &  0.18 & -0.31 &  0.35 
 \end{bmatrix}
\end{array}
\]

\section{Jointly modeling and independent univariate modeling of outcomes}\label{SM: joint vs uni}
We offer a brief discussion on jointly modeling outcomes and independent univariate modeling of outcomes. We use the setting in our first simulation experiment in Section~\ref{subsec: sim_1} of the main manuscript. We fit a latent NNGP model using the R package \textit{spNNGP} \citep{spNNGP} for each outcome individually for the simulated data there. The priors for regression coefficients and decay were the same as that of the BLMC model. The priors for the partial sill and nugget were $\mbox{IG}(2,1)$ and $\mbox{IG}(2, 0.5)$, respectively. The maximum number of nearest neighbors was set to be $m = 10$. The posterior inference was based on an MCMC chain with 5,000 iterations after an initial burn-in of 5,000 iterations. Table~\ref{table:sim2_sm} compares the posterior inference along with performance metrics of the extended data analysis with that of the BLMC model. 

Jointly modeling all the outcomes, as discussed in Section~\ref{chp4sec: Intro} of the main manuscript, often yields better predictions compared to modeling each outcome separately. In this experiment, we observed that the BLMC model provided more precise predictions than the univariate latent NNGP model based on RMSPEs, MSELs and INTs. For example, the RMSPE for the second outcome underwent a reduction of 8.5\% in the joint model as compared to the independent model. All performance metrics for measuring the posterior inferences on latent processes favor multivariate modeling than independent univariate models in this simulation study. %In summary, this study empirically shows the advantages of multivariate modeling over univariate modeling.

\begin{table}[!ht]
\caption{Simulation study summary table: posterior mean (2.5\%, 97.5\%) percentiles}
%\begin{adjustbox}{width=0.8\textwidth}
\centering
\begin{minipage}[t]{\textwidth} % <--- new
\centering
\scalebox{0.95}{
	\begin{tabular}{c|c|cc|cc}
	\hline\hline
			&  & \multicolumn{2}{c}{BLMC} & \multicolumn{2}{c}{univariate latent NNGP} \\
	\hline
	& true & inference & MCSE  & inference & MCSE \\
			$\bbeta_{11}$ & 1.0 & 0.705 (0.145, 1.233) & 0.034 &  
			0.764 (0.372, 1.199) & 0.025 \\
			$\bbeta_{12}$ & -1.0 &  -1.24 (-1.998, -0.529) & 0.045 & 
			-1.101 (-1.511, -0.596) & 0.027 \\
			$\bbeta_{21}$ & -5.0 & -4.945 (-5.107, -4.778) & 0.002 & 
			-4.96 (-5.133, -4.795) & 0.003 \\
			$\bbeta_{22}$ & 2.0 & 1.979 (1.78, 2.166) & 0.004 & 
			1.975 (1.777, 2.168) &  0.004  \\
			$\bSigma_{11}$ & 0.4 & 0.346 (0.283, 0.409) & 0.002 & 
			0.361 (0.303, 0.421) & 0.002 \\
			$\bSigma_{12}$ & 0.15 & 0.133 (0.072, 0.194) & 0.003 & 0.0 & -- \\
			$\bSigma_{22}$ & 0.3 & 0.29 (0.198, 0.386) & 0.004 & 0.299 (0.208,  0.392) &  0.004 \\
			$\phi_1$ & 6.0 & 8.723 (4.292, 14.065) & 0.343 & 
			9.393 (4.906, 13.976) &  0.247 \\
			$\phi_2$ & 18.0 & 22.63 (15.901, 29.555) & 0.416 & 
			14.086 (10.114, 18.366) & 0.226 \\
			\hline
			%MAE & & & & \\ 
			RMSPE$^a$
			%\footnotemark[1] 
			&-- & [0.728, 0.756, 0.742] &  & [0.733, 0.826, 0.781] & \\
			MSEL$^b$ 
			%\footnotemark[2] 
			& -- & [0.136, 0.168, 0.152] &  & [0.139, 0.172, 0.156] &  \\
			CRPS$^a$
			%\footnotemark[1] 
			& -- & [-0.412, -0.423, -0.418] &  & [-0.41, -0.427, -0.418] &  \\
			CRPSL$^b$
			%\footnotemark[2] 
			& -- & [-0.035, -0.038, -0.036] &  & [-0.21, -0.235, -0.222] & \\
			CVG$^a$
			%\footnotemark[1] 
			&-- & [0.915, 0.955, 0.935] &  & [0.945, 0.96, 0.9525] & \\
			CVGL$^b$ 
			%\footnotemark[2] 
			&-- & [0.946, 0.962, 0.954]  &  & [0.787, 0.798,0.792] &  \\
			INT$^a$
			%\footnotemark[1] 
			&--& [3.378, 3.756, 3.567] &  & [3.396, 4.083, 3.739]& \\
			INTL$^b$
			%\footnotemark[2] 
			& --& [0.282, 0.329, 0.305] &  &  [1.75, 1.917, 1.834] & \\
			time(s) &
			%\footnotemark[3] 
			& 143 & & 139 &
			%\footnotemark[4]
			\\
			\hline\hline
	\end{tabular}
}
\footnotetext[1]{[response 1, response 2, all responses]}
\footnotetext[2]{intercept + latent process on 1000 observed locations for [response 1, response 2, all responses]} 
%\footnotetext[3]{95\% confident interval coverage rate of intercept + latent process on 1000 observed locations for [response 1, response 2, all responses]}

\footnotetext[3]{[time for MCMC sampling, time for recovering predictions]}
\end{minipage}
\label{table:sim2_sm}
\end{table}

\section{Maps of predictions for $10$ responses of the factor BLMC model in Real Data Analysis
}\label{SM: predict_map}
\begin{figure}[!ht]
     \subfloat[ NDVI \label{subfig1}]{%
       \includegraphics[width=0.3\textwidth]{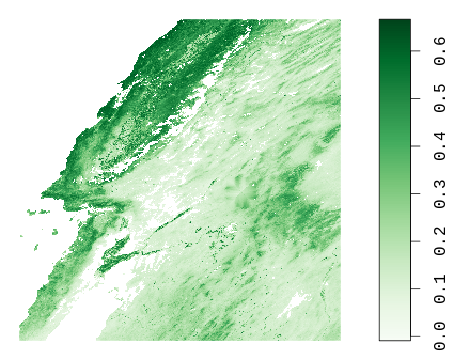}
     }
     \hfill
     \subfloat[ EVI \label{subfig3}]{%
       \includegraphics[width=0.3\textwidth]{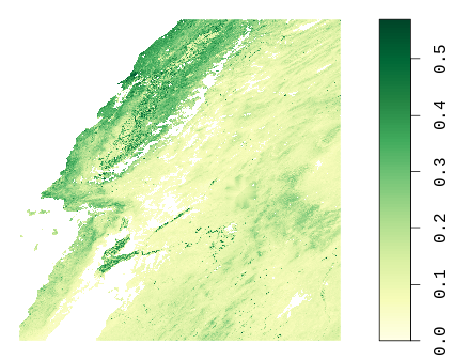}
     }
     \hfill
     \subfloat[ GPP \label{subfig5}]{%
       \includegraphics[width=0.3\textwidth]{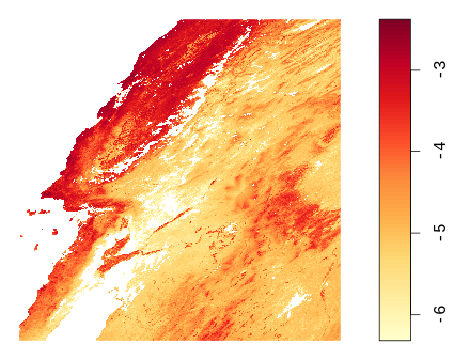}
     }
     \\
     \subfloat[ PsnNet \label{subfig7}]{%
       \includegraphics[width=0.3\textwidth]{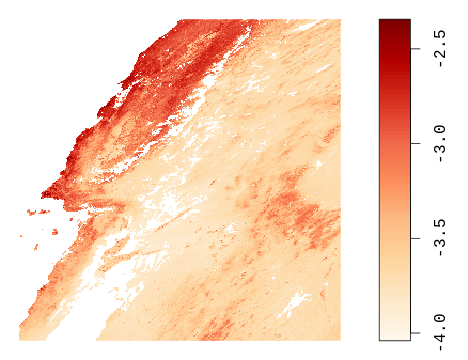}
     }
     \hfill
     \subfloat[ red refl \label{subfig9}]{%
       \includegraphics[width=0.3\textwidth]{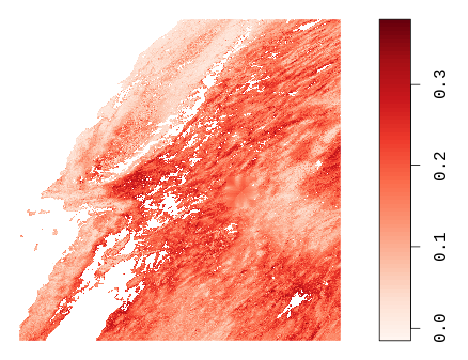}
     }
     \hfill
     \subfloat[ blue refl\label{subfig11}]{%
       \includegraphics[width=0.3\textwidth]{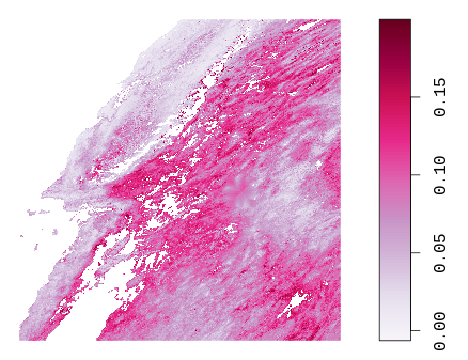}
     }\\
     \subfloat[ LE \label{subfig13}]{%
       \includegraphics[width=0.3\textwidth]{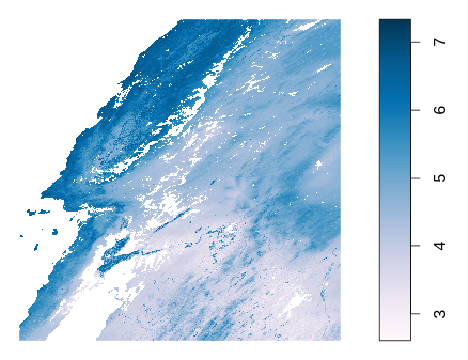}
     }
     \hfill
     \subfloat[ ET \label{subfig15}]{%
       \includegraphics[width=0.3\textwidth]{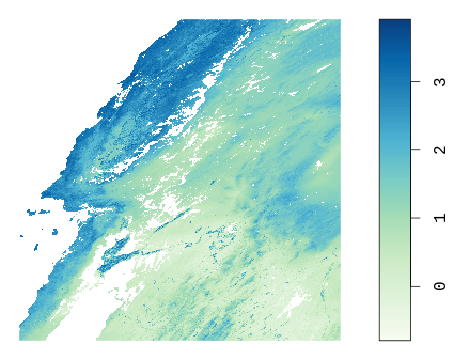}
     }
     \hfill
     \subfloat[ PLE \label{subfig17}]{%
       \includegraphics[width=0.3\textwidth]{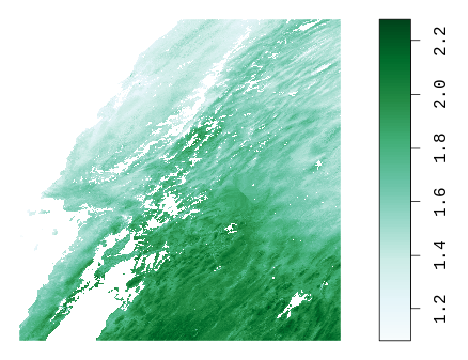}
     }\\
     \hfill
     \subfloat[ PET \label{subfig19}]{%
       \includegraphics[width=0.3\textwidth]{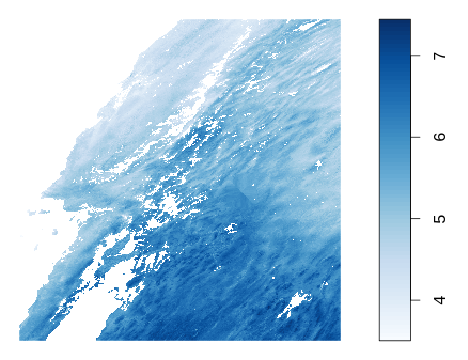}
     }
     \hfill
     \caption{Maps (a)-(j) of predicted value on $1,020,000$ observed locations for $10$ variables in Section~\ref{sec: real_data_analy}. The deeper the color, the higher the value. Some variables are transformed for better model fitting. All values are estimated by posterior mean. Each map has its own color scale. }
\label{fig:predict_maps}
\end{figure}

\label{lastpage}

\end{document}